\documentclass[a4paper,USenglish,numberwithinsect]{lipics}

\usepackage{microtype}

\usepackage{amssymb}

\usepackage{graphicx}

\usepackage{amsmath}

\usepackage{color}

\newcommand{\A}{\mathcal{A}}

\theoremstyle{plain}
\newtheorem{proposition}[theorem]{Proposition}

\title{Almost Tight Bounds for  Conflict-Free Chromatic Guarding of Orthogonal Art Galleries}
\titlerunning{Bounds for Conflict-Free Guarding}
\author[1]{Frank Hoffmann}
\author[1]{Klaus Kriegel}
\author[1]{Max Willert}
\affil[1]{Freie Universit\"at Berlin, Institut f\"ur Informatik, 14195 Berlin, Germany\\
  \texttt{\{hoffmann,kriegel,willerma\}@mi.fu-berlin.de}}
\authorrunning{F.\ Hoffmann, K.\ Kriegel, and M.\ Willert}
\Copyright{Frank Hoffmann, Klaus Kriegel, and Max Willert}
\subjclass{F.2.2 Nonnumerical Algorithms, G.2.2 Graph Theory}
\keywords{Orthogonal polygons, art gallery problem, 
hypergraph coloring}

\begin{document}

\maketitle

\begin{abstract}
\ \ \ \ We address recently proposed chromatic versions of the
classic Art Gallery Problem. Assume a 
simple  polygon $P$ is guarded by a finite set of point guards and each
guard is assigned one of $t$ colors. Such a chromatic guarding is said to be
conflict-free if each point $p\in P$ sees at least one guard with a
unique color among all guards visible from $p$. The goal is to establish
bounds on the function $\chi_{cf}(n)$ of the number of colors sufficient
to guarantee the existence of a conflict-free chromatic guarding for any
$n$-vertex polygon.

 B\"artschi and Suri showed  $\chi_{cf}(n)\in O(\log n)$ (Algorithmica,
2014) for simple orthogonal polygons and the same bound applies to
general simple polygons (B\"artschi et al., SoCG 2014).\\
In this paper, we assume the r-visibility model
instead of standard line visibility. Points $p$ and $q$ in an orthogonal polygon  are r-visible to
each other if the rectangle spanned by the points is contained in $P$.
For this model we show  $\chi_{cf}(n)\in O(\log\log n)$ and
$\chi_{cf}(n)\in \Omega(\log\log n /\log\log\log n)$.\\
Most interestingly, we can show that the lower bound proof extends to guards with line visibility. 
To this end we introduce and utilize a novel discrete combinatorial structure called multicolor tableau. This is the first non-trivial lower bound for this problem setting.

Furthermore, for the strong chromatic version of the problem, where all
guards r-visible from a point must have distinct colors, we prove a
$\Theta(\log n)$-bound.
Our results can be interpreted as coloring results for special
geometric hypergraphs.

\end{abstract}

\section{Introduction}

\ \ \ \ The classic Art Gallery Problem (AGP) posed by Klee in 1973 asks for the
minimum number of guards sufficient to watch an art gallery modelled by
an $n$-sided simple polygon $P$. A  guard sees a point in $P$ if the
connecting line segment is contained in $P$. Therefore, a guard
watches a star polygon contained in $P$ and the question is to cover $P$
by a collection of stars with smallest possible cardinality. The answer
is $\lfloor \frac{n}{3}\rfloor$ as shown by Chv\'atal, \cite{Ch}. This result was
the starting point for a rich body of research about algorithms,
complexity and combinatorial aspects for many variants of the original
question. Surveys can be found in the seminal monograph by O'Rourke
\cite{ORourke}, in Shermer \cite{Sh} or Urrutia \cite{Ur}. 

Graph  coloring arguments have been frequently
used  for proving worst case combinatorial bounds for art
gallery type questions starting with Fisk's proof \cite{Fisk}. Somehow
surprisingly, chromatic versions of the AGP have been
proposed and studied only recently. There are two chromatic variants: strong
 and conflict-free chromatic guarding  of a polygon $P$. In  both
versions we look for a guard set $G$ and give each guard one of $t$
colors. The chromatic guarding is said to be strong if for each point
$p\in P$ all guards $G(p)$ that see $p$ have pairwise different colors
\cite{ELV}.  It is conflict-free if in  each $G(p)$  there is at least
one guard with a unique color, see \cite{BS}. The goal is to determine guard
sets such that the   number of colors sufficient  for these purposes is
minimal.  Observe, in both versions minimizing the number of guards is not
part of the objective function.  Figure \ref{cf_st_example} shows a simple orthogonal polygon with
both conflict-free and strong chromatic guardings in the r-visibility model.
\begin{figure}
\centering
\includegraphics[scale=0.3]{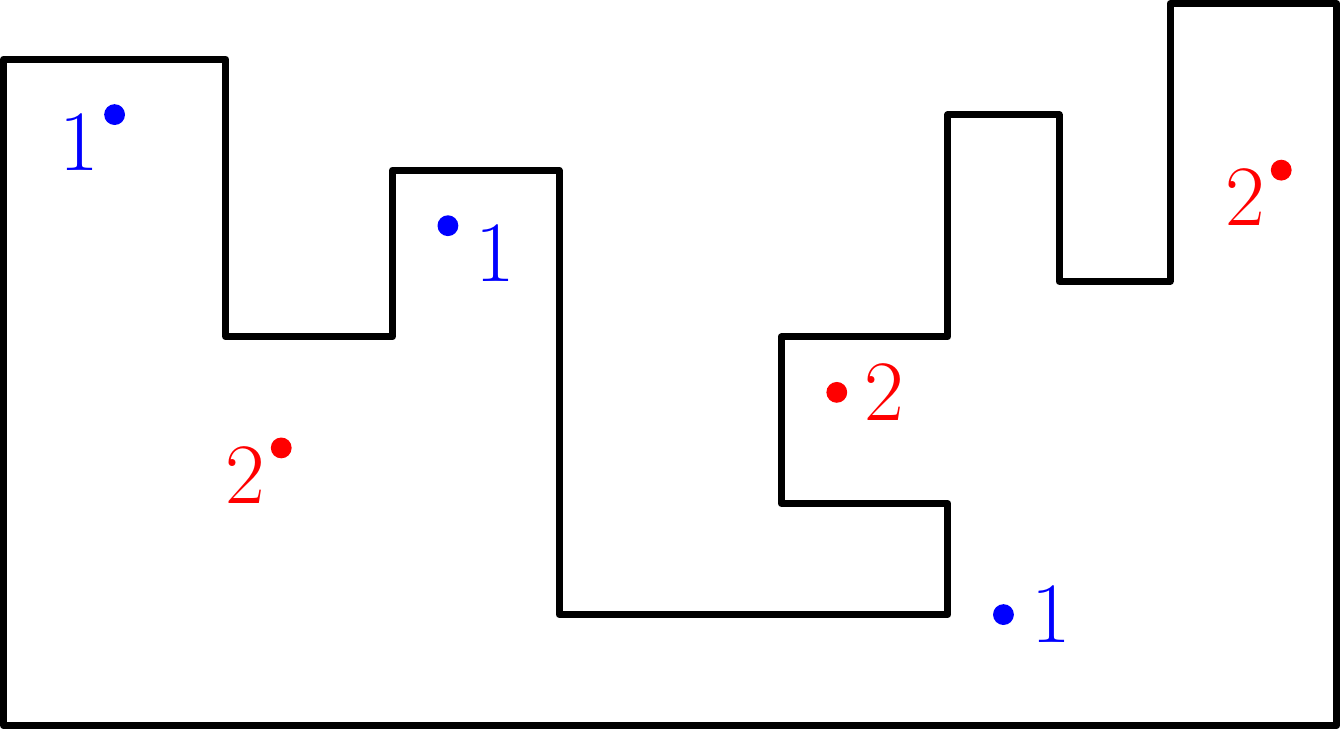}
\ \ \ \ \ \ \includegraphics[scale=.3]{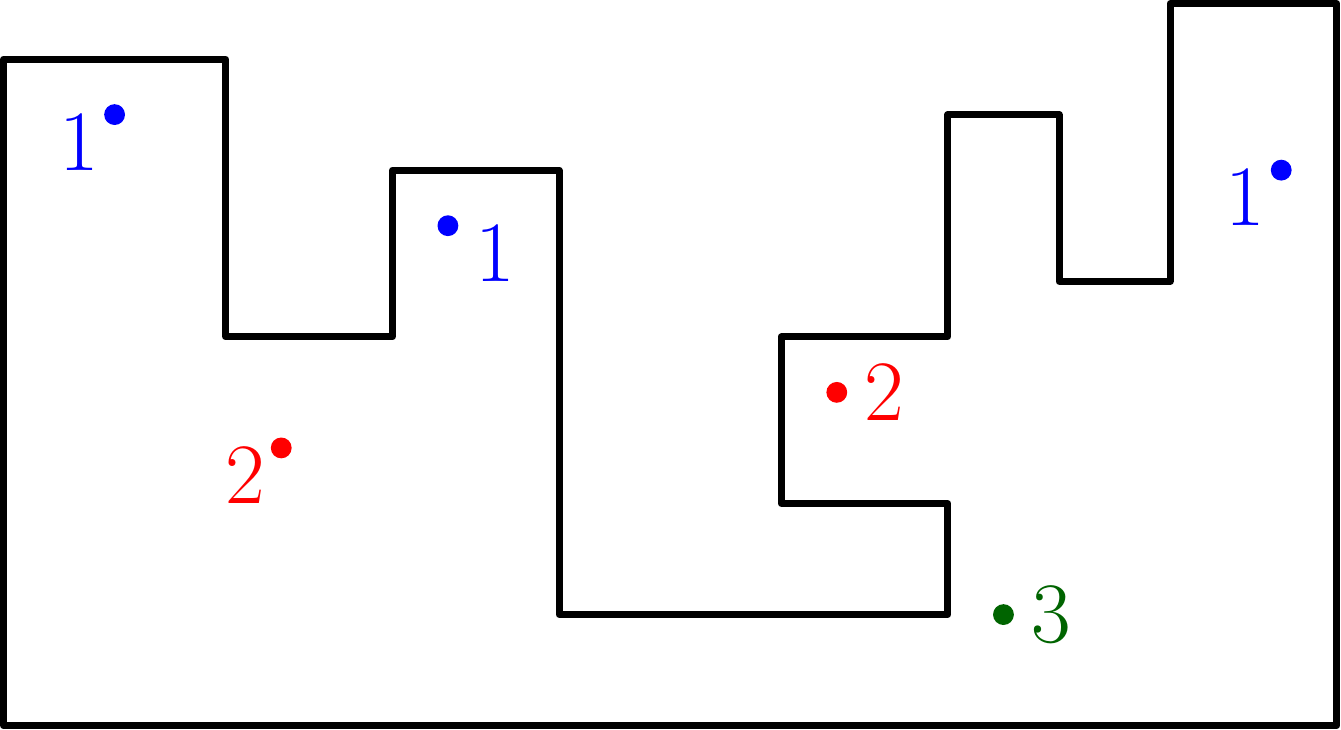}

\caption{Example of conflict-free (left) and strong chromatic (right) r-guarding}
\label{cf_st_example}
\end{figure}
To grasp the nature of the problem, observe that it has two conflicting aspects. 
We have to guard the polygon but at the same time we want the guards to hide from each other, 
since then we can give them the same color. For example, in the strong version we want a guard set 
that can be partitioned into a minimal number of subsets and in each subset the pairwise link distance is at least 3.
Moreover, we will see a strong dependence of the results on the underlying visibility model, l-visibility vs. r-visibility. We use superscripts
$l$ and $r$ in the bounds to indicate the model.

Let $\chi_{st}^l(n)$ and  $\chi_{cf}^l(n)$ denote the minimal number of colors sufficient for any simple polygon on $n$ vertices in the strong chromatic and in the conflict-free version if based on line visibility.

Here is a short summary of known bounds. For simple orthogonal polygons on $n$ vertices
$\chi_{cf}^l(n)\in O(\log n)$, as shown in  \cite{BS}. The same bound
applies to simple general polygons, see \cite{B_etal}.
Both proofs are based on subdividing the polygon into weak visibility
subpolygons that are in a certain sense independent with respect to
cf-chromatic guarding.
\\
For the strong chromatic version we have $\chi^l_{st}(n) \in \Theta(n)$
for simple polygons and $\chi^l_{st}(n) \in \Omega(\sqrt{n})$ even for
the monotone orthogonal case,
 see \cite{ELV}. NP-hardness is dicussed in \cite{FFH}. In \cite{ELV}, simple $O(1)$ upper bounds are shown
for special polygon classes like spiral polygons and orthogonal
staircase polygons combined with line visibility. 

Next we state our main  contributions for simple orthogonal polygons:\vspace*{0.2cm}
\begin{enumerate}
\item We show   $\chi^r_{cf}(n)\in O(\log\log n)$ and 
$\chi^r_{cf}(n)\in \Omega(\log\log n /\log\log\log n)$.
\item The lower bound holds for line visibility, too: $\chi^l_{cf}(n)\in
\Omega(\log\log n /\log\log\log n)$. This is the first super-constant
lower bound for this problem.
\item  For the strong chromatic version  we have $\chi^r_{st}(n)\in
\Theta(\log n)$.
\end{enumerate}
\vspace*{0.2cm}
\ \ \ \ The chromatic AGP versions  can be easily interpreted as  hypergraph coloring 
questions. Smorodinsky \cite{Smo} gives a nice survey of both
practical and theoretical aspects of hypergraph coloring.
A special role play hypergraphs that arise in geometry. For example,
given a set of points $P$ in the plane and a set of regions $\mathcal R$
like rectangles, disks etc. we can define the hypergraph $H_{\mathcal
R}(P)= (P, \{P\cap S|S\in{\mathcal R}\})$. The discrete interval
hypergraph $H_{\mathcal I}$  is a concrete example of such a hypergraph:
We take $n$ points on a line and all possible intervals as regions. It
is not difficult to see that $\chi_{cf}(H_{\it I}) \in\Theta (\log n)$.
As to our AGP versions, we can associate with a given polygon and a guard set a geometric hypergraph.
Its vertices are the guards and a hyperedge is defined by a set of guards that have a nonempty common intersection of their visibility regions and in the intersection there is
a point that sees exactly these guards.  Then one wants to color this graph in a conflict-free or in a strong manner. \\
Another  example is the following  rectangle hypergraph.
Vertex set is a finite set of $n$ axis-aligned rectangles and each
maximal subset of rectangles with a common intersection forms a
hyperedge. Here
the order for the cf-chromatic number is $\Omega(\log n)$ and
$O(\log^2n)$ as shown in \cite{Smo,PaTo}.

Looking at our results,  it is  not a big surprise that the 
combination of orthogonal polygons with r-visibility yields the
strongest bounds. This is simply due to additional structural properties and this  phenomenon has already been observed for the original AGP. For example, the $\lfloor\frac{n}{4}\rfloor$ tight
worst case bound for covering simple orthogonal polygons with general
stars can also be proven for r-stars (see \cite{ORourke}) and it
holds even for orthogonal polygons with holes, see (\cite{Hoff}).
Further, while minimizing the number of guards is NP-hard both for
simple general and orthogonal polygons if based on line visibility, it
becomes polynomially solvable for r-visibility in the simple orthogonal
case, see \cite{MRS,Wor}. The latter result is based on the solution of
the strong perfect graph conjecture.

The paper is organized as follows. We give neccesary basic definitions in the next section. Then we prove upper bounds in Section 3 using techniques developed 
in \cite{BS, B_etal}. Our main contribution are the lower bound proofs in Section 4. Especially, we introduce a novel combinatorial structure called multicolor tableau.
 This structure enables us to extend the lower bound proof for r-visibility
to the line visibility model.

Omitted proofs can be found in the Appendix.

\section{Preliminaries}
\subsection{Orthogonal polygons, r-visibility and general position assumption}
\ \ \ \ We study simple orthogonal polygons, i.e., polygons   consisting of
alternating  vertical and horizontal edges only that do not have holes.
By $|P|$ we denote the number of vertices, by  $\partial P$ the boundary
and by $\mbox{int}P= P\setminus \partial P$ the interior of the polygon.
Vertices can be reflex or convex. A reflex vertex has an interior angle 
$3\pi/2$ while convex vertices have an interior angle of $\pi/2$. To
simplify the presentation we make the following very weak assumption
about general position of orthogonal polygons: If two reflex vertices
$p,q\in P$ are connected in $\mbox{int}P$ by a horizontal/vertical chord
then the four rays emanating from $p$ and $q$ towards the interior along
the incident edges represent only 3 of the 4 main compass directions.
That is, two rays are opposite to each other and the other two point in
the same direction. 

Points $p,q\in P$ are line visible (or l-visible for short) to each
other if the line segment $pq$ is containd in $P$. Observe that the
segment $pq$ is allowed
to contain parts of boundary edges. The points $p,q$ are r-visible to
each other if the closed axis-parallel rectangle $R[p,q]$ spanned by the points is contained
in $P$. 
For $p\in P$ we denote by $V^l_P(p)=\{q\in P| pq\subset P\}$ and
$V^r_P(p)=\{q\in P| R[p,q]\subseteq P\}$ the set of all points l-visible
from $p$  and r-visibility, respectively. This is also called the
visibility polygon of a point $p\in P$. If it is clear from the context
which polygon is meant we omit the index.  A polygon that is fully
visible from one of its points is called a star and, again, we have to
distinguish between l-stars and r-stars. Most notably, for a point $p$
in an orthogonal polygon the visibility polygon $V^r(p)$
is itself orthogonal while $V^l(p)$ usually is not. We can generalize
this by defining
for a subpolygon $P'\subset P$ its visibility polygon by
$V^r(P')=\cup_{p\in P'}V^r(p)$. The windows of a subpolygon $P'$ in $P$
are those parts of  $\partial P'$
 that do not belong to $\partial P$.

For an orthogonal polygon $P$  we define its induced  r-visibility line
arragement $A^r(P)$.
 Two points $p,q\in P$ are equivalent with respect to r-visibility if
$V^r(p)=V^r(q)$.
This is an equivalence relation. What are the equivalence classes? 
First of all, there is a simple geometric construction to find $A^r(P)$.
For each reflex vertex of $P$ we extend both incident boundary edges
into  $\mbox{int}P$ until they meet the boundary again, therefore
defining a subdivision of the polygon. The faces of this line
arrangement are rectangles, line segments, and intersection points.
Clearly, two points from the interior of the same rectangle define the
same r-star. What about line segments in the arrangement? We extend a
line segment $l$ into both directions until we hit a convex vertex or
the interior of a boundary edge. Let's call this extension $l^+$. By our
general position assumption we know that on one side
 ({\it inner} side) of $l^+$ there is only polygon interior.  Consider a
point $p$ in the interior of a line segment that is incident with two
rectangular faces. It is not difficult to see, that $p$    inherits the
r-visibility from the incident rectangle on its inner side and the same
rule applies to intersection points which can have up to four incident
rectangles.

Finally, we define special classes of orthogonal polygons. A weak
r-visibility polygon (also known as histogram) has a boundary edge $e$
(called base edge) connecting two convex vertices such that
$V^r_P(e)=P$. This is therefore a monotone polygon with respect to  the
orientation of $e$. A weak r-visibility polygon that is an r-star is
called a {\it pyramid}.

\subsection{Conflict-free and strong chromatic guarding}
\ \ \ \ A set $G$ of points is an r-guard set for an orthogonal polygon $P$ if their r-visibility polygons
jointly  cover the whole polygon. That is: $V^r(G)=\cup_{g\in
G}V^r(g)=P$, analogously for l-visibility. If in addition  each guard
$g\in G$ is assigned one color $\gamma(g)$ from a fixed finite set of
colors $[t]=\{1,2,\ldots ,t\}$ we have a chromatic guarding $(G,\gamma)$.
Next we give the central
definition of this paper.

\begin{definition} A chromatic r-guard set $(G,\gamma)$ for $P$ is  strong
 if for any two  guards $g,g'\in G$ we have
$V^r(g)\cap V^r(g')\not =\emptyset$ implies $\gamma(g)\not =\gamma(g')$.\\
 A chromatic r-guard set $(G,\gamma)$  is  conflict-free   if for any
point $p\in P$ in the guard set $G(p)=V^r(p)\cap  G$ there is at least one
guard with a unique color. 
\end{definition}
\ \ \ \ We denote by $\chi^r_{cf}(P)$ the minimal $t$ such that there is
conflict-free chromatic guarding set for $P$ using $t$ colors. Maximizing this value
over  all polygons with $n$ vertices from  a specified polygon class is
denoted by $\chi^r_{cf}(n)$.\\
Consequently, we  denote by $\chi^r_{st}(P)$ the minimal $t$ such that there is
strong chromatic guarding set using $t$ colors. Maximizing this value
for all polygons with $n$ vertices from  a specified polygon class defines
the value  $\chi^r_{st}(n)$.\\
The notions for line visibility are completely analogous and use superscript $l$.

\section{Upper Bounds }
\ \ \ \ We show  upper bounds for both strong and conflict-free r-guarding of
simple orthogonal polygons of size $n$: $\chi_{st}^r(n)\in O(\log n)$ and 
$\chi_{cf}^r(n)\in O(\log\log n)$. These bounds are even  realized by r-guards
placed in the interior of visibility cells. This restriction will
simplify the arguments.  The proof (see also \cite{Wi}) follows closely ideas
developed in \cite{B_etal, BS} for conflict-free
l-guarding of simple  polygons. Therefore we only recall the
general ideas, omit some proof details and emphasize the differences
stemming from the underlying r-visibility.

\subsection{Partition into independent weak visibility polygons}
\ \ \ \ First of all, we reuse the central concept of {\it independence} introduced
in \cite{BS, B_etal} for line
visibility. Independence means that one can use the same color sets for
coloring guards in independent subpolygons.
 The following definition  suffices for our purposes.

\begin{definition}
Let $P$ be a simple orthogonal polygon and  $P_1$ and $P_2$
subpolygons of $P$.
 We call $P_1$ and $P_2$  independent if  $V^r( \mbox{int}P_1)\cap V^r(\mbox{int}
P_2)=\emptyset$.
\end{definition}

Next, we are going to subdivide hierarchically an orthogonal polygon $P$
into  weak visibility subpolygons by a standard window partitioning
process as  described in \cite{BS}. \\
Remark: In the following we use the term subdivision not in the strong set-theoretic sense. A subdivision of  $P$ into closed subpolygons $P_1,\ldots ,P_k$ means that 
$P=\cup_{i=1}^kP_i$ and for all $i\not= j$ we have $\mbox{int} P_i\cap \mbox{int} P_j=\emptyset$. \\
The subdivision is
represented by a tree
${\mathcal T}={\mathcal T}_P(e)$ with the weak
 visibility polygons as node set. Let  $e$ be a highest horizontal edge
of $\partial P$, the ``starting'' window.  $Q=V^r(e)$ is a weak
visibility polygon and is the root vertex  of $\cal{T}$. Now $Q$ splits
$P$ into parts and defines a finite set (possibly empty if $Q=P$) of
vertical windows $w_1,\ldots w_k$. Each window  corresponds to a left or
right turn of a shortest orthogonal path from $e$ to the subpolygon
lying entirely behind the window. Then we recurse, see Figure \ref{partitionFigures}.\\
By the partitioning process we can obtain a linear number of subpolygons only. There are $\frac{n-4}{2}$ reflex vertices in a simple orthogonal polygon
 with $n$ vertices. Each window uses at least one reflex vertex. Therefore,
we get at most $n/2-1$ weak visibility polygons. This bound is realized for example by spiral polygons.

\begin{figure}
\centering
\includegraphics[scale=.3]{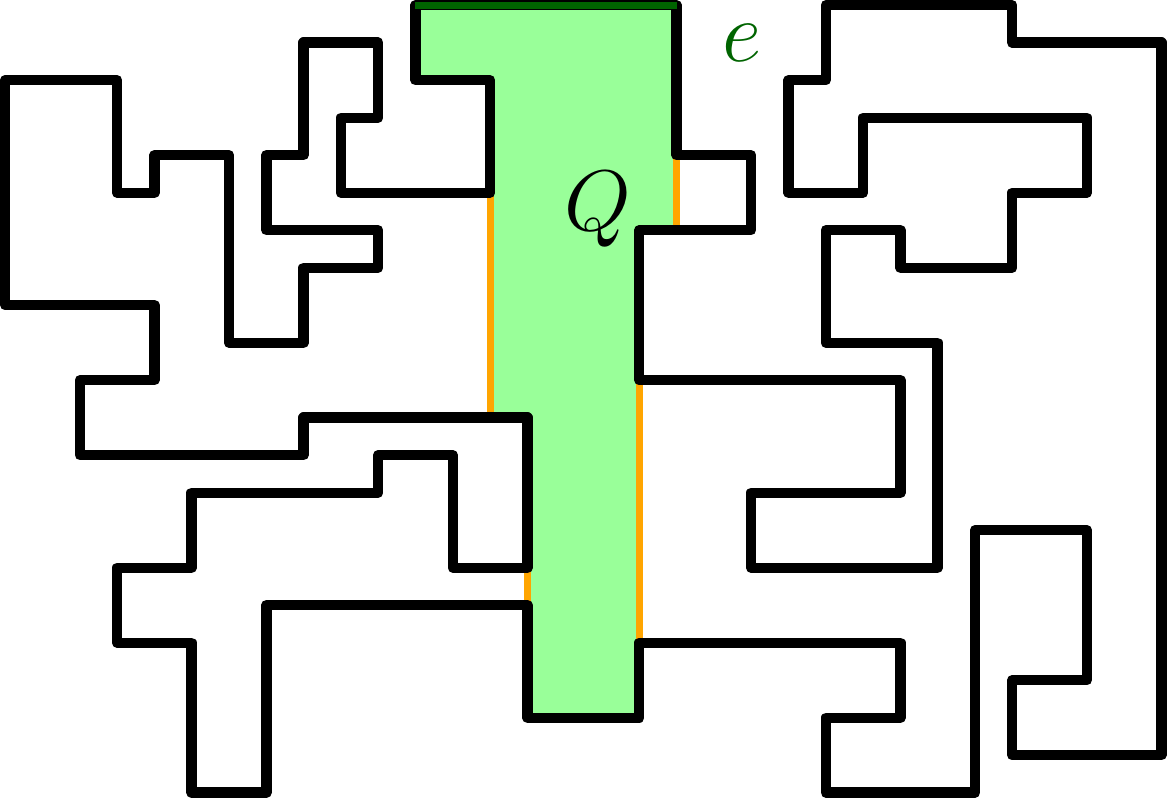}\quad\quad
\includegraphics[scale=.3]{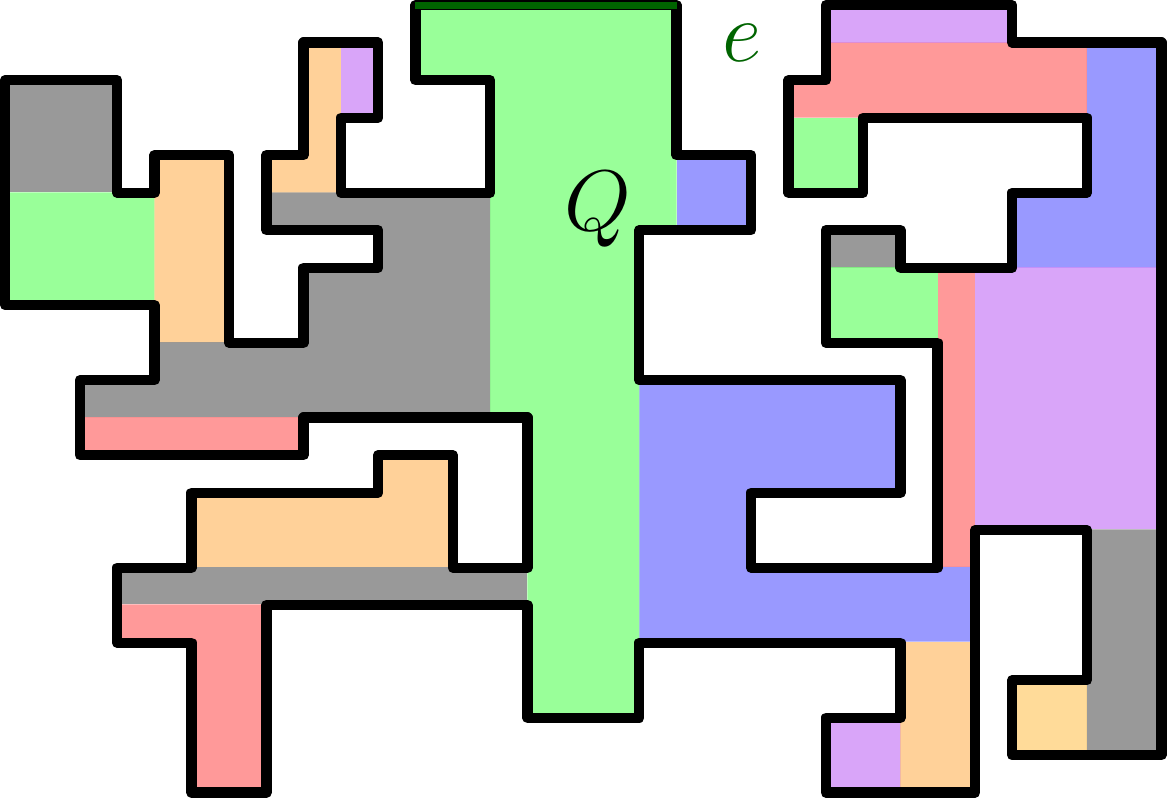}\quad\quad
\includegraphics[scale=.4]{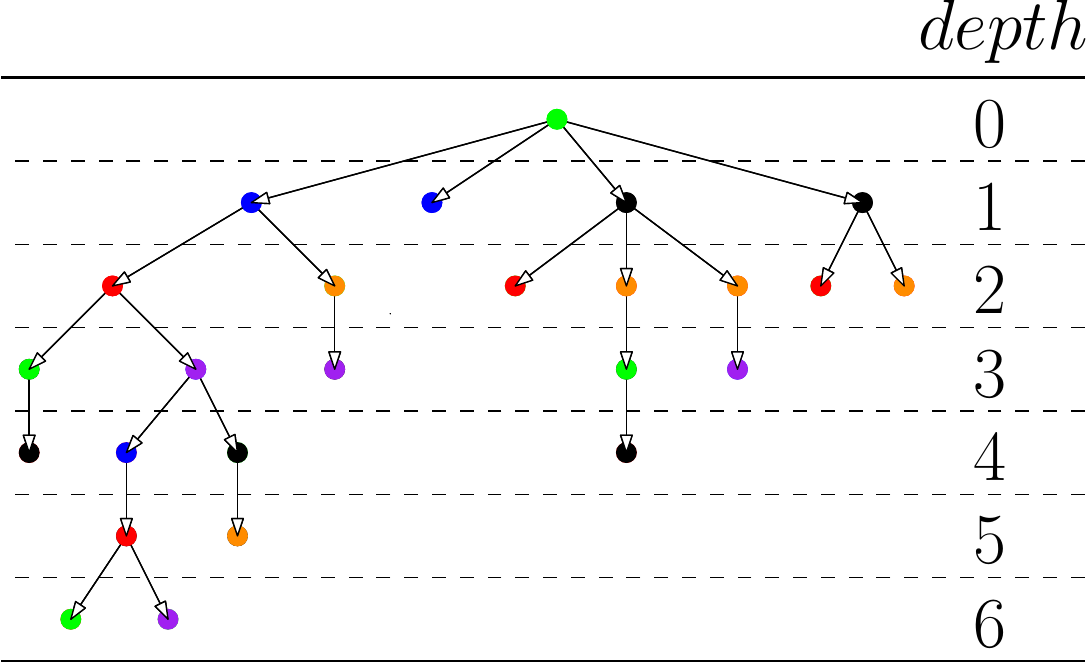}
\caption{The partitioning process and the corresponding schematic tree.}
\label{partitionFigures}
\end{figure}
Let $\A_d, d=0,1,2$ be the family of all weak visibility polygons
corresponding to nodes of depth congruent $d\mod 3$ in $\cal{T}$. We
partition  $\A_d$ into $\A_d^L$ consisting of $Q$ and all those
subpolygons which are left
 children and, on the other side,  $\A_d^R$ consisting of the remaining ``right'' parts.

\begin{lemma}
\label{indep}
Let $P$ be a polygon and  $\A^L_d , d=0,1,2$ the family of subpolygons corresponding to  left nodes in
$\cal{T}$ with depth congruent $d \mod 3$. Then the interior of  subpolygons in 
$\A^L_d$ have  pairwise link distance at least three, analogously for $\A^R_d$.
\end{lemma}
\begin{proof}
Suppose there are two different  subpolygons $P_1$ and $P_2$ in 
$\A^L_d$. If  they have different depth then for arbitrary points
$p_1\in int P_1$ and $p_2\in int P_2$ any orthogonal path connecting
these points has
length
 at least 3. Otherwise they have the same depth. In this  case they
could be sibling nodes with parent node $P_0$. To walk orthogonally 
from $p_1$ to $p_2$ it needs  two parallel edges to cross the windows
plus one more edge in
 $P_0$. If the lowest common ancestor $P_0$ is more than 1 level above
then a shortest orthogonal path from $p_1$ to $p_2$ has to visit the
parent node of $P_1$, the parent node of $P_2$ and then descend to $p_2$
which takes at least three edges.
\end{proof}

Observe that distinguishing left and right nodes is essential. It can be
possible to walk with one step from a left node to a right sibling.

\begin{corollary} Let $P_1,P_2\in \A^L_d$ be subpolygons computed in the subdivision process for $P$.
Then $P_1$ and $P_2$ are independent and there exists a strong chromatic r-guarding for $P$ in which guards in $P_1$ and $P_2$ use the same
color set. The same is true for conflict-free chromatic guarding.
\end{corollary}

\begin{proof} 
Assume we have  r-guards $p_i$ in $\mbox{int} P_i$ , $i=1,2$ with $V^r(p_1)\cap
V^r(p_2)$ containing a point $q$. Then there exist points $p_1'\in
R[p_1,q]\cap (\mbox{int} P_1)  , p_2'\in R[p_2,q]\cap (\mbox{int} P_2)$ and a
connecting orthogonal path
$p_1'-q-p_2'$ of length 2. But  this contradicts the previous lemma.
Therefore  $P_1$ and $P_2$ are independent and strong chromatic r-guardings for $P_1$ and $P_2$ do
not interfere with each other,
the same holds for conflict-free r-guardings.
\end{proof}
Remark: We will restrict the guards to sit in the interior of visibility cells. However, this does not effect the asymptotic upper bounds on the number of colors used.
\subsection{Guarding  a weak visibility polygon}

\ \ \ \ Consider  a weak visibility polygon $P$ with a horizontal base edge $e$.
An edge of $P$ opposite to $e$ is an $r$-edge if it connects two reflex vertices, it is a
$c$-edge if it has two convex vertices. Among the horizontal edges opposite
to $e$ there is at least one $c$-edge and  a chain connecting two consecutive 
$c$-edges contains exactly one $r$ edge. Recall that a pyramid $P$
contains exactly one $c$-edge $e_1$. We guard a pyramid with one r-guard
stationed opposite to  $e_1$
 in the visibility cell just below base edge $e$. Next we describe (see
\cite{BS}) a simple truncation process that decomposes a weak
visibility problem into pyramids.

{\bf The truncation process:}
Let $P$ be a weak visibility polygon with $n$ vertices, a  horizontal  base
edge $e$ on top of the polygon and $C$ the set of $c$-edges opposite to $e$. If there is only
one such $c$-edge we stop and return $P$. Otherwise, for each $c$-edge $e'$ we
sweep $P$ from $e'$ 
towards $e$ until the sweep line reaches the first neighboring  $r$-edge.  
We truncate $P$ by cutting of the pyramid below. After processing all
edges in $C$ we have again a weak visibility polygon $P^{(1)}$ with base
edge $e$. Observe that $P^{(1)}$ does not depend on the order in which
we process the edges in $C$
 and, moreover, the pyramids associated with $C$ are independent.\\
Then we iterate with $P^{(1)}$ and get $P^{(2)}$ and so on. Eventually, we have indeed
partitioned $P$ completely into pyramids. These pyramids have an
important structural property. By construction, the unique $c$-edge in
each pyramid contains a non-empty segment of the original boundary
$\partial P$, we call them ``solid'' segments. As guard position for
such a pyramid we choose an interior point  just below the base edge of
the pyramid opposite to an interior point of a solid segment.  
 
In Figure \ref{covGuards} we see an example of a weak visibility
polygon, its decomposition into pyramids and the chosen guard positions.
Again, there is a canonical schematic tree representing the
decomposition and the guard positons.

\begin{figure}
\centering
\includegraphics[scale=.3]{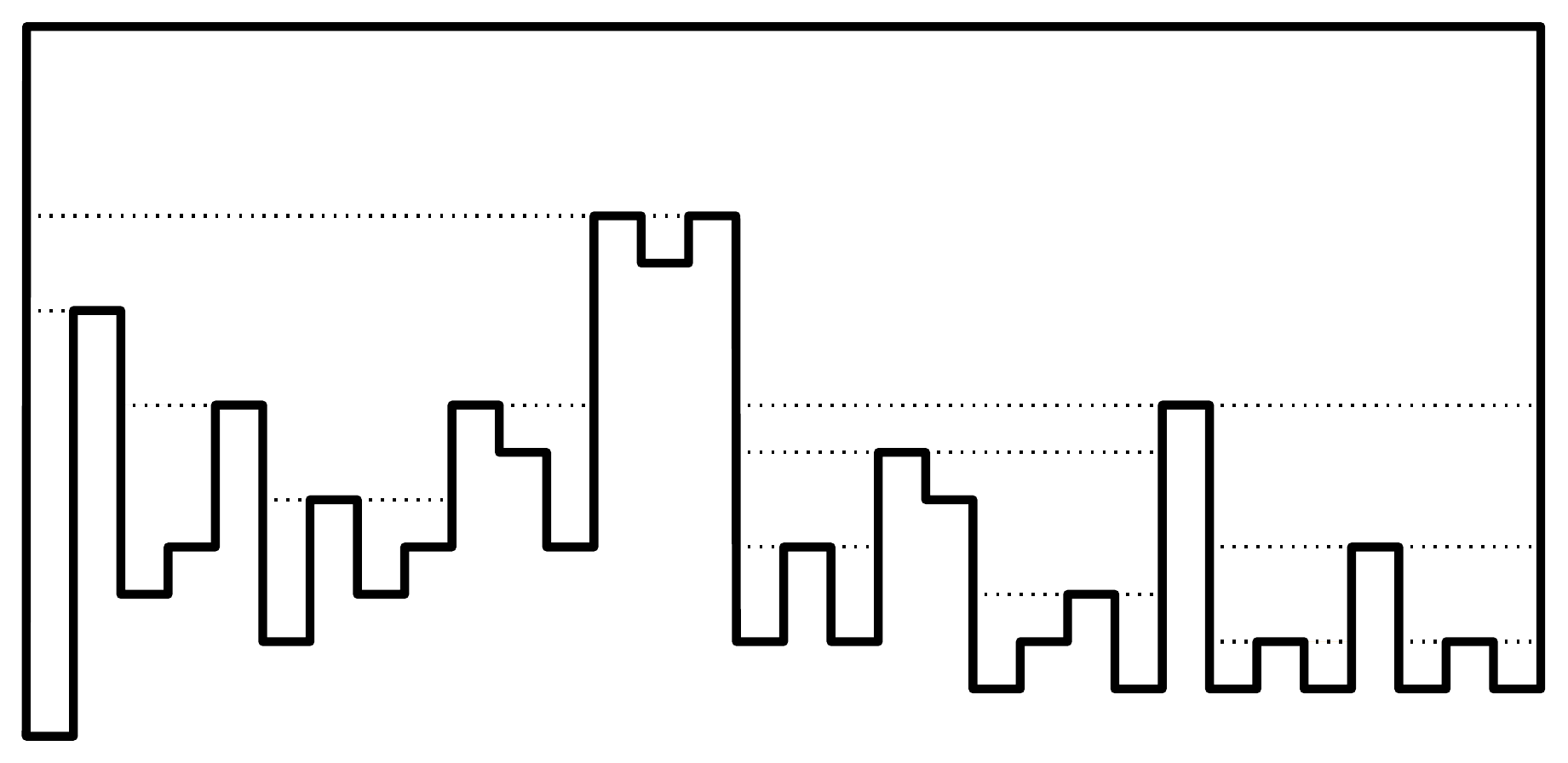}\quad\quad
\includegraphics[scale=.3]{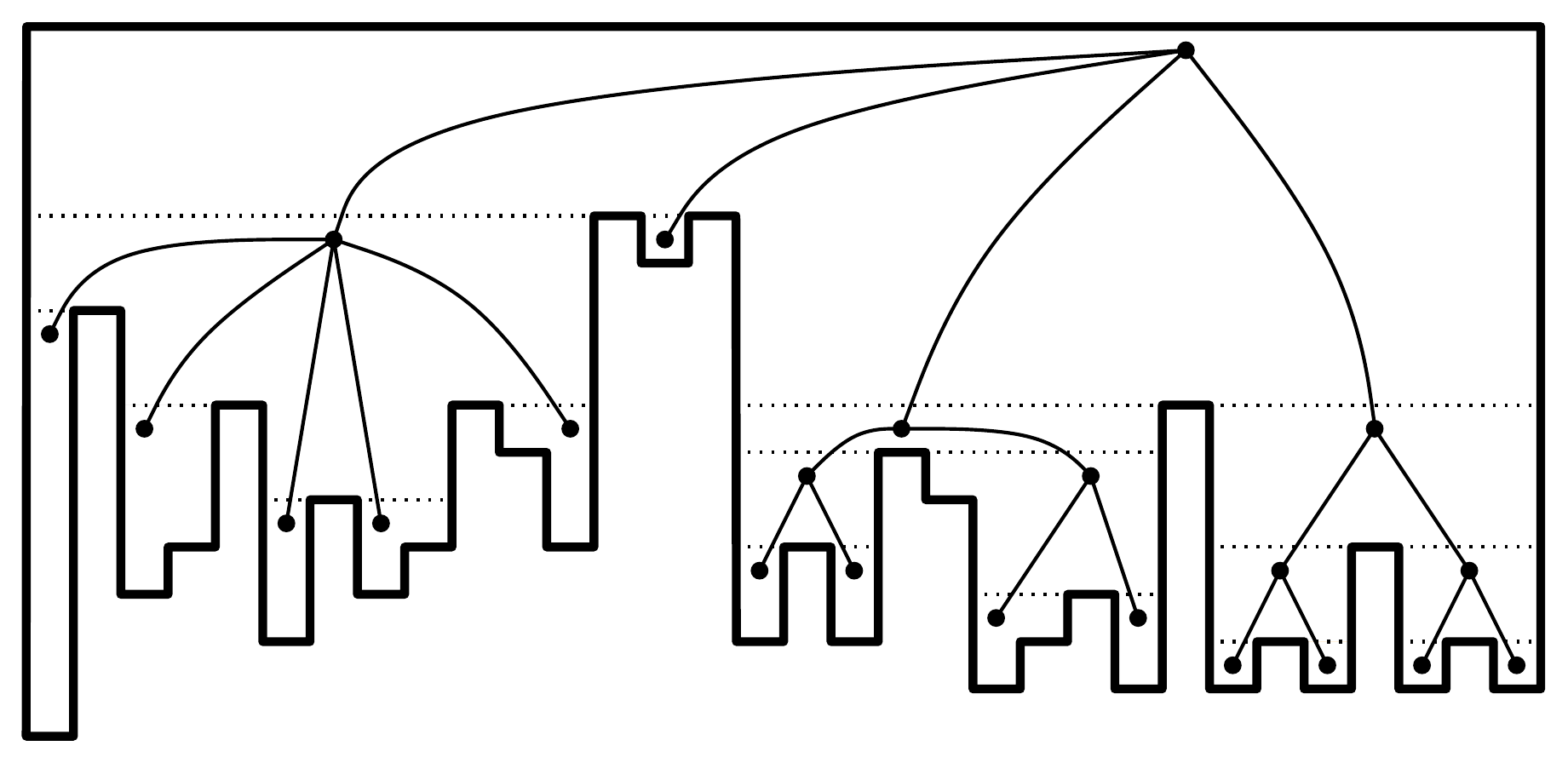}
\caption{The truncation  process of a  weak visibility polygon and its
schematic tree with guard positions}
\label{covGuards}
\end{figure}

Clearly, the height of $T_P, |P|=n$ is in $O(\log{n})$. In the worst case, this is best possible as shown by the spike
polygons $S_m$ in Section 4  we use for our lower
bound proofs.

\begin{figure}
\centering
\includegraphics[scale=.3]{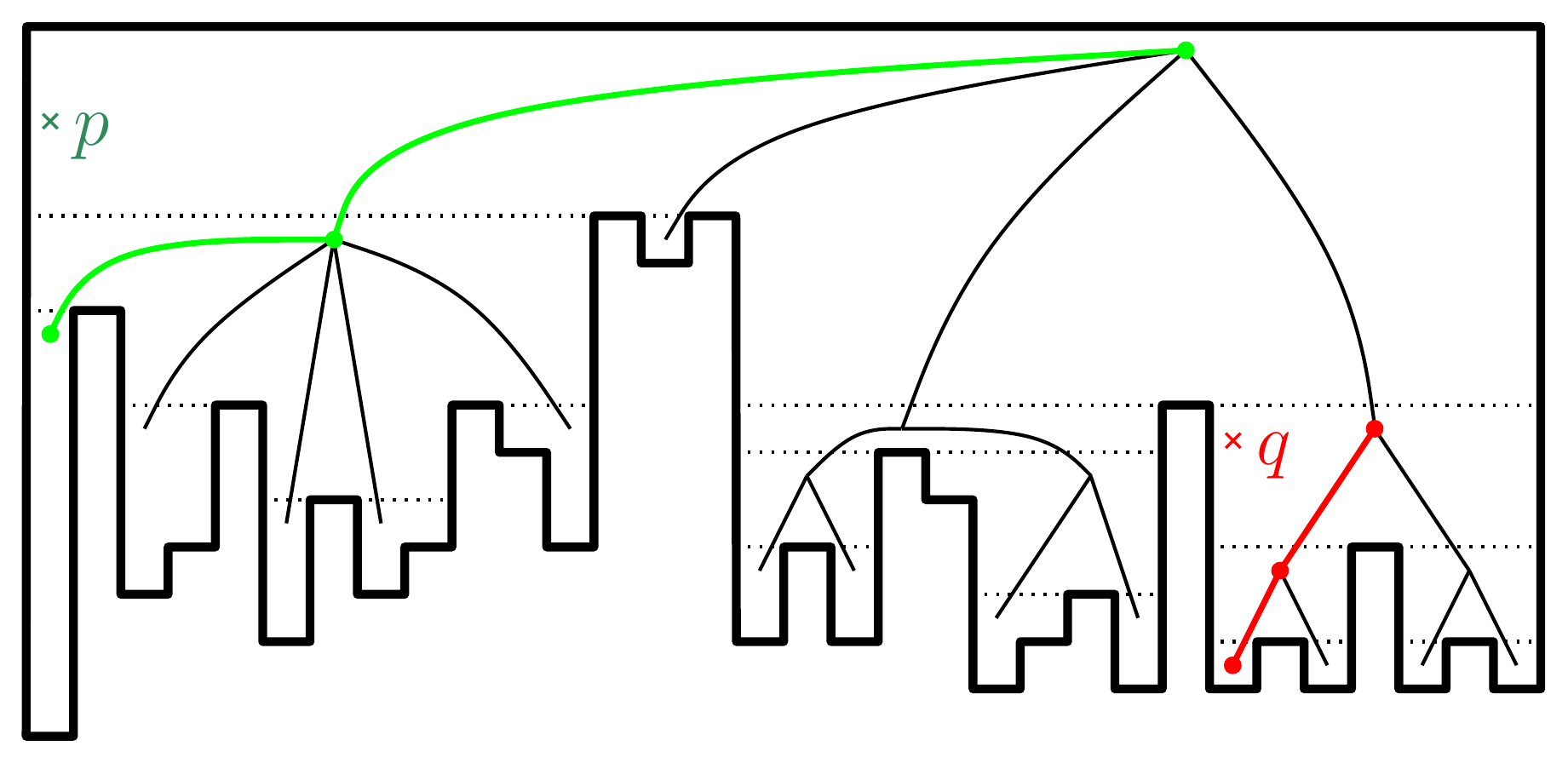}
\caption{Single points are seen by connected  chains of r-guards}
\label{guardTree}
\end{figure}

The following lemma states the main structural property for this  tree
of guards.
\begin{lemma}
\label{chain}
Let $P$ be  a weak visibility polygon and $T_P$ the
guard-tree computed in the truncation process. Then for each $p\in P$
all guards in $G(p)$  form  a connected subpath of a  root-to-leaf
path in $T_P$.
\end{lemma}

\begin{proof}
Assume two nodes representing pyramids $P_1$ and $P_2$ are not on a root
to leaf path in $T_P$. Consider the lowest common anchestor node, say it
is pyramid $P_0$. $P_0$ has a solid segment in its c-edge. Therefore
$P_1$ and $P_2$ are independent and r-guards from both pyramids cannot
see the same point $p$. Now we know that all r-guards watching a common
point $p$ are indeed on a common root-to-leaf path. Let $g_l$ be the
deepest and $g_h$ the highest guard among them with $g_h\not= g_l$. We
have to show, that all guards in between see point $p$, too. Where can
point $p$ be? It has to be in the vertical strip above the base line of
the pyramid with guard $g_l$ and below the base line of the pyramid
corresponding to $g_h$, since the parent node of $g_h$ does not see $p$
by assumption. This region is a rectangle $R$. For any guard $g$ between $g_l$
and $g_h$ the vertical strip above the corresponding base line contains
$R$ and $g$ sees $p$.  
\end{proof}
In Figure \ref{guardTree} the paths formed by r-guards watching point $p$ and for point $q$
are indicated.

\begin{theorem}
\label{stTheorem}
Let $P$ be an orthogonal  polygon with $|P|=n$. We have
$\chi^r_{st}(P)\in O(\log{n})$.

\end{theorem}

\begin{proof} We decompose $P$ into pairwise independent weak visibility
polygons. Each weak visibility polygon can be further decomposed into 
pyramids and the corresponding guard trees have height 
$O(\log{n})$. We color each guard by its depth in the tree.
 This is a strong chromatic guarding since for each $p\in
P$ by Lemma \ref{chain} all  of its guards have pairwise different colors.
\end{proof}

We use the same r-guard positions but a different coloring scheme to get a
conflict-free coloring. Consider the color alphabet $[m]=\{1,2,\ldots
,m\}$ and the following recursively defined set of words. Let $s_1=1$
and $s_i=s_{i-1}\circ i\circ s_{i-1}$. The following is straightforward
and has been used before for conflict-free coloring the discrete
interval hypergraph.
\begin{lemma}
\label{sequence}
A prefix of $s_m$ with length $k$ has no more than
$\lceil\log(k+1)\rceil$ different colors and each connected subword contains a unique color.
\end{lemma}

\begin{theorem}
\label{cfTheorem}
Let $P$ be an orthogonal polygon with $|P|=n$. Then 
 $\chi^r_{cf}(P)\in O(\log{\log{n}})$.
\end{theorem}

\begin{proof}
The only difference in comparison with the proof above is the coloring
scheme. Each r-guard tree gets colored top-down with the sequence $s_m$ of
length at most height of the tree, that is $O(\log n)$. By Lemma
\ref{sequence} the color alphabet needs to be of size
$O(\log{\log{n}})$ and the coloring is conflict-free by Lemma \ref{chain}.
\end{proof}

We illustrate the construction in Figure \ref{coveringGuards}.
\begin{figure}
\centering
\includegraphics[scale=.3]{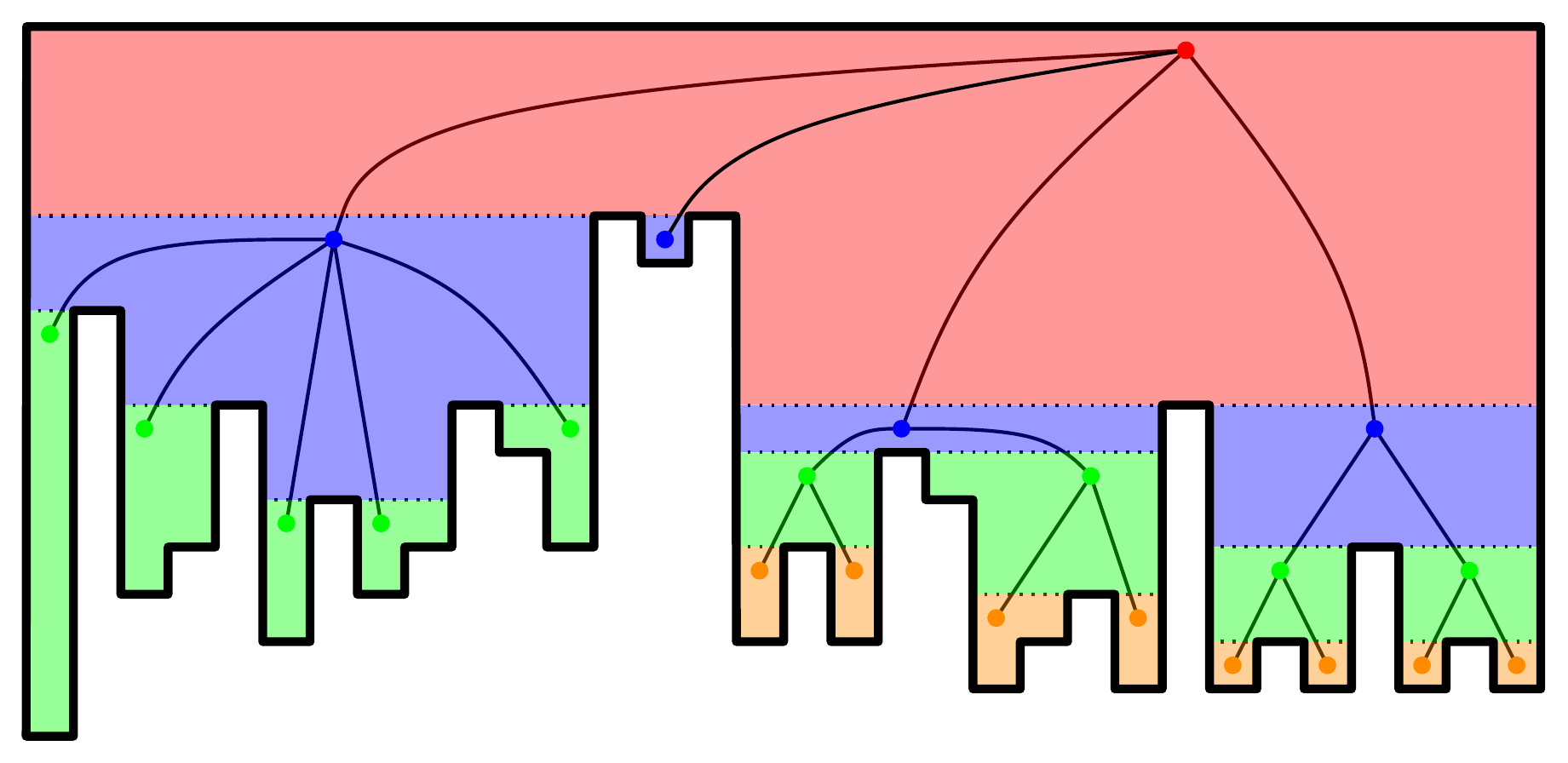}\quad\quad
\includegraphics[scale=.3]{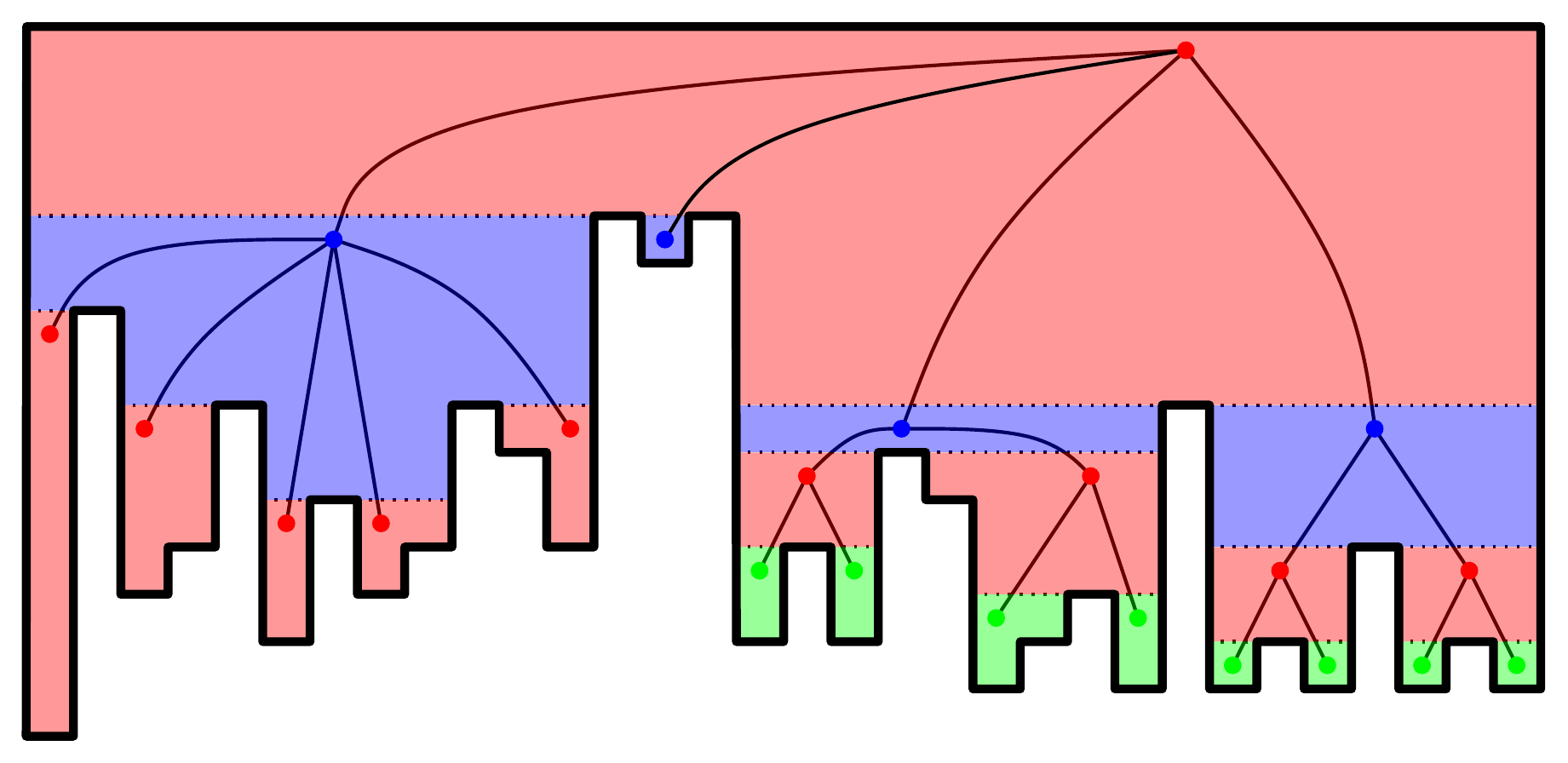}
\caption{Strong (left)  and conflict-free guarding (right) of a  weak visibility
polygon}
\label{coveringGuards}
\end{figure}

\section{Lower Bounds}

\subsection{Spike polygons}
\ \ \ \ All lower bounds established in this paper are based on a simple,
recursively defined family of so called
spike polygons $S_m$, where $S_1$ is a simple square and $S_{m+1}$ is
formed by two copies of
$S_m$ separated by a vertical spike, but joined by an additional
horizontal layer.
The left side of Figure \ref{spikePoly}  illustrates this construction
together
with the subdivision of $S_2$ into visibility cells. Observe that the
height sequence of the
spikes in $S_{m+1}$ is nothing else but the word $s_m$ used in Theorem \ref{cfTheorem} above.

Columns of $S_m$ are numbered  left to right by indices $k \in
[2^m -1]$,
and  cells in column $k$ top down by an addditional index $i \in
[d_m(k)]$ where
$d_m(k)$ is the depth of column $k$ in $S_m$. Formally, we have
$d_m(k)=m-\pi_2(k)$  where $\pi_2(k)$ is the multiplicity of  factor
$2$ in the prime decomposition of $k$.  
Obviously, a column has maximal depth $m$ iff its index is odd.

We introduce the notions of the left and right wing of  column $k$ in
order to distinguish
guard positions: The left wing $W_L(k)$ is the set of all points
strictly on the left side of the midline of column $k$  and the
right wing is  the complement $W_R(k) = S_m \setminus W_L(k)$.

\begin{figure}
\centering
\includegraphics[scale=.35]{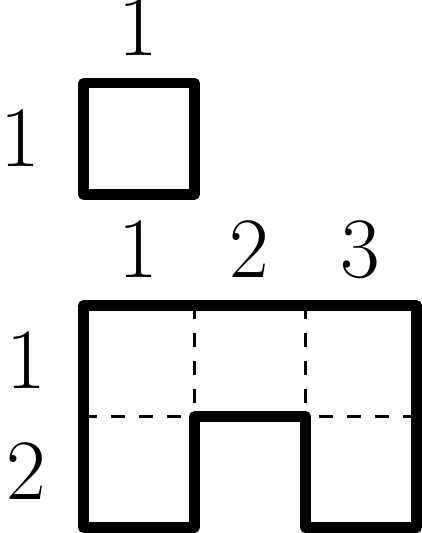} \qquad
\includegraphics[scale=.35]{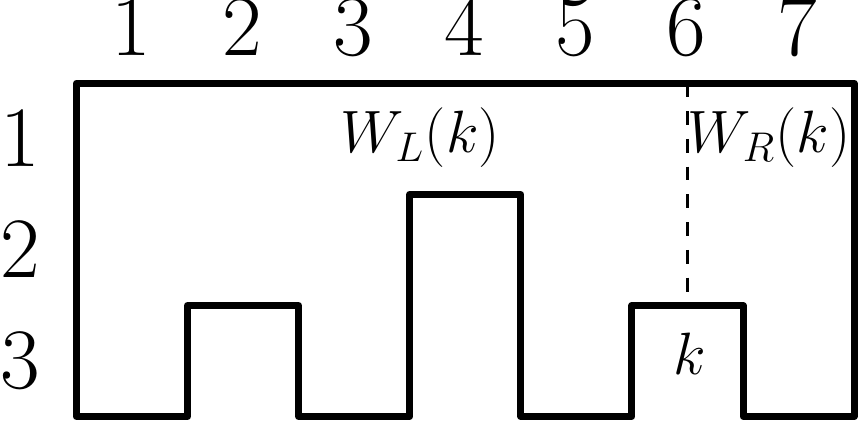}       \qquad
\includegraphics[scale=.35]{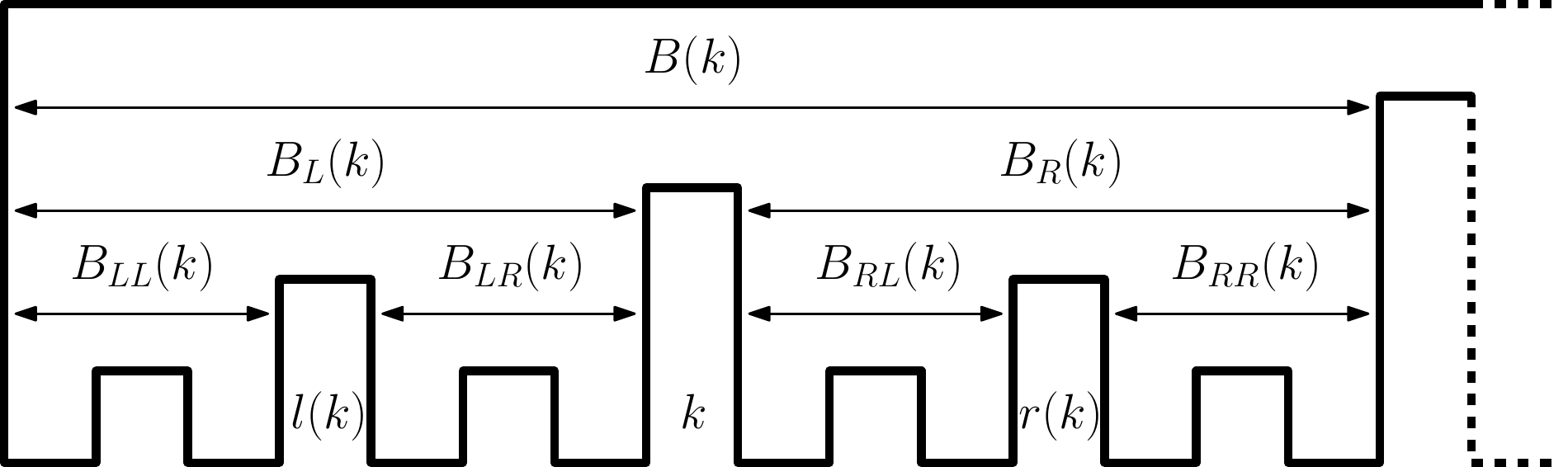}
\caption{Spike polygons $S_1$ and $S_2$ (left), left wing and right wing
of column $k=6$ in $S_3$ (middle), blocks and
  subblocks (right)}
\label{spikePoly}
\end{figure}

We will prove three lower bound results for guarding spike polygons. 
%The common idea is
%an inductive argument starting with the fact, that it is impossible to
%guard
%the polygon $S_2$ conflict free with only one color.
The easiest version refers to strong chromatic
r-guardings.

\noindent
\begin{theorem}
We have $\chi_{st}^r(S_t) \geq t$.
\end{theorem}
\begin{proof} The proof is by induction. The induction base for $S_1$ is straightforward.
Next we show the induction step by contradiction. Assume that the claim is true for some  $S_t$  and suppose that there is  a
strong chromatic r-guarding of $S_{t+1}$ with $t$ colors only.
There must be a unique color $c_1$ for the top cell in the middle
column. Since the
corresponding guard $g_1$ sees all cells in the first row, it is
the only one of color $c_1$ in $S_{t+1}$
(any other $c_1$-guard  would produce a conflict in at least one cell in
the first row).
The deletion of the  top row splits the remaining part of  $S_{t+1}$
into two copies of $S_t$.
Depending on the position of $g_1$ in $S_{t+1}$, at least in one  copy  no cell is
$r$-visible from $g_1$.
Thus, we have  a strong $(t-1)$-chromatic r-guarding  of this copy.
But this contradicts
 the induction hypothesis.\end{proof}

The other two lower bound proofs  are much more involved, but they follow the same scheme. They are by induction and the induction step is shown by contradiction. 
But now the  induction step  
can require a sequence of $t$  steps cutting out from  the original $S_m$
 smaller units 
until arriving at a contradiction. We start with the proof  for
cf-guardings
with respect to $r$-visibility. In its quintessence  
 it  relies on purely combinatorial properties of a  discrete structure which we call multicolor tableau.
We will then rediscover a slightly  weaker version of this structure having similar  properties when discussing lower bounds for conflict-free guardings
based on  l-visibility
for appropriately vertically stretched spike polygons.

\subsection{Blocks and multicolor tableaux}
\ \ \ \ Consider the spike polygon $S_m$. It has $N=2^m-1$ columns.
We define the block $B(k)$ of column $k$ as the interval
of all neighbouring columns of depth at least $d(k)$, see Figure \ref{spikePoly}: \\
 $B(k)= \left[k-\left(2^{\pi_2(k)}-1 \right), k+  \left(2^{\pi_2(k)}-1 \right)\right]$. Geometrically, a block is nothing but a smaller spike polygon. Deleting its central column a block splits into a left and a right subblock: \\
$$
 B_L(k)= \left[k-\left(2^{\pi_2(k)}-1 \right), k-1\right] \ \
 B_R(k)= \left[k+1, k+  \left(2^{\pi_2(k)}-1 \right)\right]$$
For odd $k$ we have $B(k)=\{k\}$ and $B_L(k) = B_R(k)= \emptyset$.
Later  it will be necessary to subdivide a left
or right subblock again into
its left and right subblocks. These ``quarter''-subblocks can be described
making use of the definition above
together with the central column $l(k)=k-2^{\pi_2(k)-1}$ in block
$B_L(k)$ and column
$r(k)=k+2^{\pi_2(k)-1}$ in  block $B_R(k)$:\\
\[
B_{LL}(k) = B_L(l(k)) \quad   B_{LR}(k) = B_R(l(k))  \quad
B_{RL}(k) = B_L(r(k))  \quad B_{RR}(k) = B_R(r(k)).
\]

Let $G$ be a finite set of $r$-guards covering $S_m$ and
$\gamma : G \rightarrow [t]$
a cf-coloring of  $G$. By  $M_{i,k}$ we denote the  multiset of
all colors of guards that see the $i$th
visibility cell $R_{i,k}$ in column $k$, and let $m_{i,k}(c)$ denote the
multiplicity of color $c$ in  $M_{i,k}$.
Then the  combinatorial scheme
${\mathcal M}(\gamma):=\left(M_{i,k} \, | \, 1 \leq k \leq N,  1 \leq i
\leq d_m(k) \right)$
will be called a conflict-free multicolor tableau.
\\
We
formally define   the set of unique colors of a cell by $U_{i,k} := \{ c \in [t] \, | m_{i,k}(c)=1\} $ and the  standard inclusion relation $M_{i,k} \subseteq M_{j,l}$ for multisets  by:
              $\forall c \in [t] \quad  m_{i,k}(c) \leq m_{j,l}(c)$.

The following simple fact  about r-visibility in spike polygons
makes  the crucial difference between
the simpler lower bound  proof for
r-visibility and the more involved proof for l-visibility.
% Observe, that we don't distinguish between a guard and its location in the polygon.

\begin{lemma}
\label{rVis}
Let $g$ be an r-guard in $S_m$ that sees a cell $R_{i,k}$,
then $g$ is  in a cell of depth $d \leq d_m(k)$
and it sees all cells $R_{i',j}$ with
$i' \leq i$ and $j \in B(k)$.
\end{lemma}
\begin{proof}
The first assertion is straightforward because otherwise the spike in
column $k$ would block
the visibility between $g$ and $R_{i,k}$. For the second claim
consider the minimal rectangle $R$ enclosing the cell of $g$ and 
$R_{i,k}$. Since the lower
side of $R$ has depth  $\leq d_m(k)$ and all columns $j \in B(k)$ have
depth  $\geq d_m(k)$
one can extend $R$ within $S_m$ horizontally to the whole width of the
block $B(k)$
and upwards to the top edge of $S_m$.
\end{proof}

\begin{lemma}
\label{uniqueInterval}
Let $G$ be an r-guard set covering $S_m$ and  $\gamma : G \rightarrow [t]$
a cf-coloring of $G$. Then for any color $c \in [t]$ and for any column in
the multicolor tableau ${\mathcal M}(\gamma)$ the following holds: The
multiplicity $m_{i,k}(c)$
is a monotonically decreasing function with respect to  row index $i$.
In particular, if $c$ is a unique color somewhere in column $k$ then
row indices of cells
with $m_{i,k}(c)=1$ form an interval $[i_1,i_2]$ and  $m_{i_2+1,k}(c)=0$.
\end{lemma}

\begin{proposition}
\label{r-visAdmiss}
For  a conflict-free r-guarding  $\gamma : G \rightarrow [t]$ of $S_m$
the multicolor
tableau ${\mathcal M}(\gamma)$
has  three combinatorial properties:
\begin{enumerate}
\item cf-Property: $\forall_{ k \in [N]} \ \forall_{   i \in [d_m(k)]} \
U_{i,k} \not= \emptyset$.
\item Monotonicity: $\forall_{ k \in [N]} \ \forall_{  1 \leq i < i'  \leq
d_m(k)} \ M_{i',k} \subseteq M_{i,k}$.
\item Left-right rule:   If $c$ is a unique color in the top cell
$R_{1,k}$ of column $k$ then for
all $j \in B_L(k)$ or for all $j \in B_R(k)$ the following three
conditions hold
\begin{enumerate}
\item $c \in M_{1,j}$
\item If $c  \in U_{1,j}$ then  $c  \not \in M_{d_m(k)+1,j}$.
\item If $c \not \in U_{1,j}$ then $c \not \in U_{1,j'}$ for all $j' \in
B(j)$.
\end{enumerate}
\end{enumerate}
\end{proposition}
\begin{proof}
There is nothing to prove for the cf-property and the monotonicity
follows from
Lemma \ref{uniqueInterval}.
It remains to establish the left-right rule.
Assume  $c \in U_{1,k}$ and consider  the corresponding
guard $g$.
Depending on whether $g$ is in $W_R(k)$ or in $W_L(k)$ we 
prove that the
three properties hold in the opposite block $B_L(k)$ or in $B_R(k)$, respectively. Again,
it suffices to
discuss the first case. By Lemma \ref{rVis} $g$ sees all cells $R_{1,j}$
with $j \in B(k)$.
Condition (a) holds because $B_L(k) \subseteq B(k)$. \\
We prove  condition (b) by contradiction assuming  that $c  \in U_{1,j}$
and $c \in M_{d_m(k)+1,j}$
for some $j \in B_L(k)$. Since $g$ is in the right wing
of $k$, it can't see any cell of depth $d_m(k)+1$ in the left wing. Thus
$c \in M_{d_m(k)+1,j}$ implies the existence of another c-colored guard
$g'$, that sees
$R_{d_m(k)+1,j} $. But again by Lemma \ref{rVis} $g'$  sees also
$R_{1,j}$ what contradicts
the uniqueness of $c$ for this cell. \\
Finally, if $c$ is not unique for  $R_{1,j}$ then there are at least two
guards with color $c$ that watch
$R_{1,j}$. Both of them watch all cells $R_{1,j'}$ with $j' \in
B(j)$ what proves
condition (c). \end{proof}

\begin{theorem}
\label{r-cf-lowerbound} For simple orthogonal polygons
$\chi_{cf}^r (n) \in
\Omega\left(\frac{\log{\log{n}}}{\log{\log{\log{n}}}}\right)$.
\end{theorem}
\begin{proof}(Sketch)
Let  $m(t)$ be defined by $m(1)=2$ and $m(t)=1+t \cdot
m(t-1)$ for
$t \geq 2$.

\noindent
{\em Claim:}
$ \ \chi_{cf}^r(S_{m(t)}) > t$.

\noindent
It is easy to  deduce the theorem from the claim. A simple inductive
argument shows 
$m(t) \leq (t+1)!$ and thus $\chi_{cf}^r(n) > t$ for some  $n \leq
2^{(t+1)!+1}$,
because this  is an upper bound on the vertex number of   $S_{m(t)}$.
This inequality is equivalent to $\log{n}  \leq (t+1)!+1$ what implies
$\log{\log{n}} \in O(t \log{t})$ and, finally,
$t \in \Omega\left(\frac{\log{\log{n}}}{\log{\log{\log{n}}}}\right)$.

We prove the claim by induction on $t$. For the base case $t=1$
we must show that it is impossible to guard $S_2$ conflict free with one
color.
Suppose the opposite and consider the corresponding multicolor tableau
$\mathcal{M}=(M_{1,1},M_{1,2},M_{1,3},M_{2,1}M_{2,3})$. The only way to
fulfill
the uniqueness condition is to set $M_{i,j}=U_{i,j}=\{1\}$ for all pairs
$(i,j)$ and color $1$. This already  contradicts  condition
(b) of the
left-right rule applied to the situation $1 \in U_{1,2}$.

Next, we illustrate the induction step in detail for the step from
 $t=1$ to $t=2$ with $m(1)=2$ and $m(2)=5$. We prove it by contradiction.
Suppose there  is an r-guard set $G$ and a
coloring $\gamma: G \rightarrow [2]$ that is a conflict-free guarding of $S_{5}$
and let $\mathcal{M}(\gamma)$ be the corresponding  multicolor tableau.
A contradiction will be derived by  a sequence of at most two
cutting stages with the goal to identify a subpolygon $S_{2}$ in $S_5$ that has a
conflict-free r-guarding with only one  color.\\
We start with a unique color $c_1 \in [2]$ of the top cell $R_{1,16}$ in the central column
$k_1=16$ of $S_{5}$. W.l.o.g.~the corresponding guard $g_1$ is located in the right wing
$W_R(16)$ and the three conditions of the left-right rule apply for all $j \in B_L(16)$. \\ 
The  subblocks $B(4)$ and $B(12)$  cover $B_L(16)$ with the exception
of the separating column $8$.  Considering the two central top cells $R_{1,4}$ and $R_{1,12}$
(the green cells in  Figure \ref{r-vis-indStep1})
 we distinguish  two cases:\\
  (1): $\ \forall_{j \in \{4,12\}}\ c_1 \in U_{1,j}$\\
  (2): $\ \exists_{j \in \{4,12\}} \ c_1 \not\in U_{1,j}$\\
\begin{figure}
\centering
    \includegraphics[width=0.6\columnwidth]{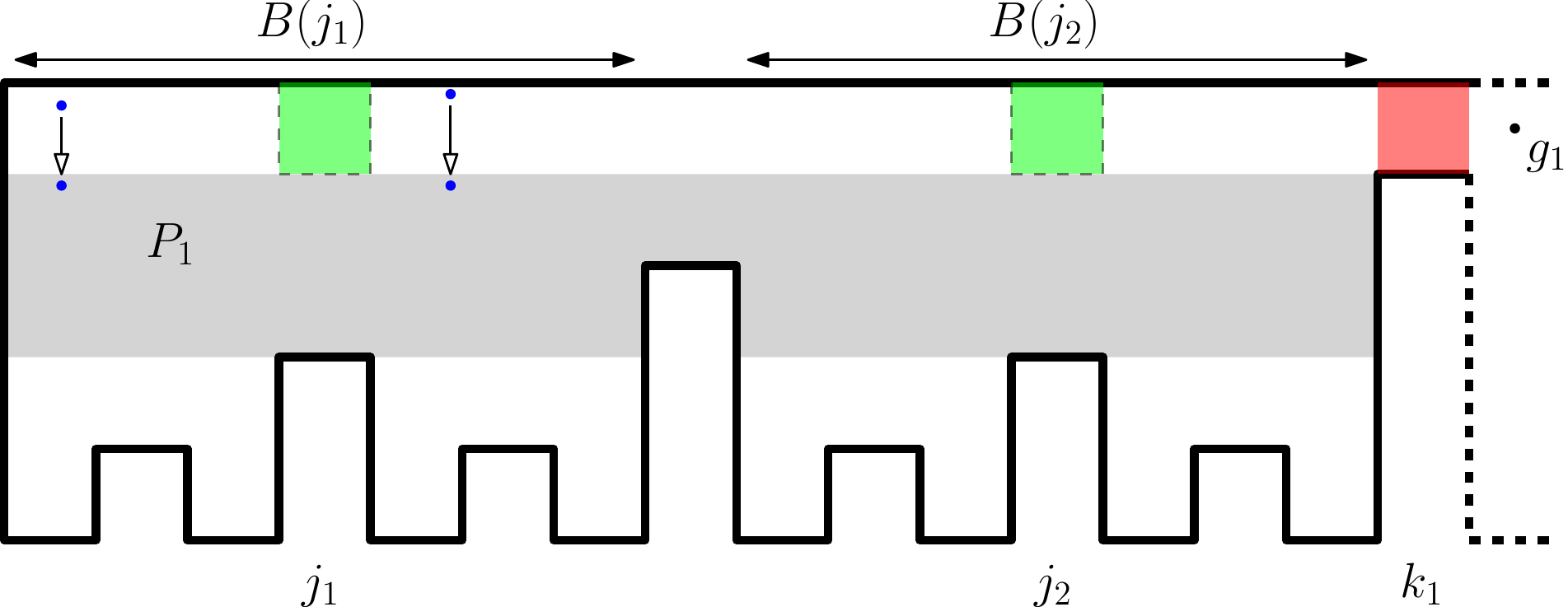}
\caption{The first subdivision stage in $S_5$: Case 1 holds if  $c_1$ is
unique for both green cells, $j_1=4$ and $j_2=12$.
Then $P_1$ would have a conflict-free r-guarding with only one color, a
contradiction.}
\label{r-vis-indStep1}
\end{figure}
Whenever Case (1) occurs  this is a stopping
rule, because  one can directly identify a subpolygon  with the
shape of $S_{2}$ together a conflict-free guarding that uses  one color only, a contradiction.
To that end we construct the subpolygon $P_1$ consisting of all cells
$R_{i,j}$ with $2 \leq i \leq 3$ and $j \in B_L(k_1)$, the  grey shaded region 
in Figure \ref{r-vis-indStep1}.
One can make two basic observations about $P_1$:
\\
(i) The shapes of $P_1$ and  $S_{2}$ are the same in the sense that 
$P_1$ is a stretched version of $S_2$ and their
decompositions  into r-visibility cells are isomorphic. 
\\
(ii) Let $G_1$ be the set of all guards from $G$ that  are positioned in
$P_1$. We extend it to a
set $G_1^{+}$ by pulling down
all guards in cells above $P_1$ onto the top edge of $P_1$. In Figure \ref{r-vis-indStep1}
 this is illustrated by small downarrows. Then  $G_1^{+}$
with the original coloring $\gamma$ is a cf-guarding of
$P_1$ with one color only because color $c_1$ does not occur. \\The last
assertion is
straightforward because the presence of any $c_1$-colored guard in
$G_1^{+}$ would contradict
the uniqueness of $c_1$ for  $R_{1,4}$  or $R_{1,12}$ (the assumption of case 1).
It is also clear that $G_1^{+}$ covers $P_1$ because any original guard
for a cell $R_{i,4}$
with $i \geq 2$  remains in $G_1^{+}$ and it will cover all $R_{i,j}$
with $j \in B(4)$ as well. 
Finally the cf-condition also extends from a cell $R_{i,4}$ to all
$R_{i,j}$ with $j \in B(4)$ because all columns in $B_L(4)$ and 
$B_R(4)$ are truncated from
below at level $3$. The  argumentation applies to cells $R_{i,j}$
with $j \in B(12)$. 
\\
Observations (i) and (ii) together give a contradiction to the
inductive assumption.\\
In contrast, the ocurrence of case (2)  invokes a second  stage.
Choose one index $j \in \{4,12\}$ such that  $c_1 \not \in U_{1,j}$, set $k_2=j$
and  repeat the former procedure in the block $B(k_2)$. Remark, the left-right rule
implies $c_1 \not \in U_{1,j}$ for all $j \in B(k_2)$.  
Now the second color $c_2$ must be unique for cell $R_{1,k_2}$ and the position
of the corresponding guard $g_2$ implies that the three conditions of the left-right rule apply 
for all $j \in B_L(k_2)$ or for all $j \in B_R(k_2)$.
\begin{figure}
\centering
    \includegraphics[width=0.6\columnwidth]{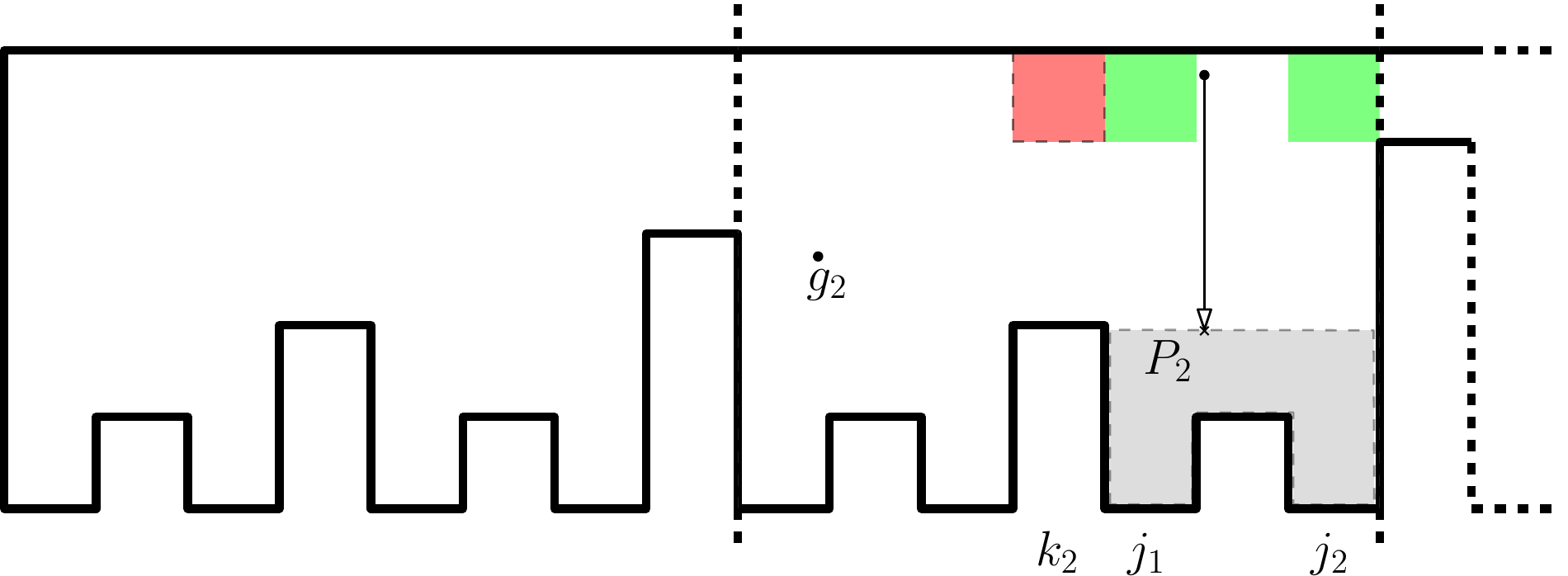}
\caption{Case 2 ocurred in the first subdivision stage because of $c_1
\not \in U_{1,k_2}$, $j_1=13$ and $j_2=15$. Then
the next subdivision stage applies to block $B(k_2)$.}
\label{r-vis-indStep2}
\end{figure}
Figure \ref{r-vis-indStep2} illustrates the situation for  $k_2=12$ and  $g_2 \in W_L(12)$.
Note, guard $g_2$ could sit also outside of block $B(12)$.
The three conditions of the left-right rule apply for all $j \in B_R(12)$. 
Again, there are two subblocks  $B(13)$ and $B(15)$ (now single columns) that cover  $B_R(12)$
with exception of the separating column $14$. The next case distinction 
refers to their  top cells $R_{1,13},R_{1,15}$: \\
  (1): $\ \forall_{j \in \{13,15\}}\ c_2 \in U_{1,j}$\\
  (2): $\ \exists_{j \in \{13,15\}} \ c_2 \not\in U_{1,j}$\\
In Case (1) one can cut out the subpolygon $P_2$ consisting of all cells $R_{i,j}$
with $4 \leq i \leq 5$ and $j \in B_R(12)$, see Figure \ref{r-vis-indStep2}, and construct 
a guard set $G_2^+$ as in case 1 before. This would result in a cf-guarding of $P_2$ without $c_2$,
a contradiction. 
\\
In  Case (2) there is a $j \in \{13,15\}$  with  $c_2 \not \in U_{1,j}$, but moreover 
$c_1 \not \in  U_{1,j}$ because Case (2) occured in the first stage. This implies $U_{1,j}= \emptyset$,
a contradiction again.
\\
Now we present the  general induction step from $t-1$ to $t$, again shown by contradiction. \\
Assume  that  there is
no conflict-free r-guarding of $S_{m'}$ with
$t-1$ colors for $m'=m(t-1)$, but there  is an r-guard set $G$ and a
coloring $\gamma: G \rightarrow [t]$ that is a conflict-free guarding of $S_{m}$
for $m=m(t)$. Again we make use of the corresponding  multicolor tableau
$\mathcal{M}(\gamma)$. \\
A contradiction will be derived by  a sequence of at most $t$
cutting stages. \\
Stage $s$ ($s \in [t]$)  always starts with the precondition that 
there is a column $k_s \in [2^m-1]$ of depth $d_m(k_s)=1+(s-1)m'$,
the block $B(k_s)$ of $S_m$ with $N_s= 2^{m-(s-1)m'}-1$ 
columns and a set $C_{s-1} \subseteq [t]$ of $s-1$ colors such that 
$c \not \in U_{1,j}$ for all $j \in B(k_s)$ and for  all $c \in C_{s-1}$.
The first stage starts with $k_1=2^{m-1}$ (the central column of $S_m$),
$N_1=2^m-1$, $C_0 = \emptyset$ and an empty precondition.
Each stage results in a case distinction where the ocurrence of the first case
would finish the proof by a contradiction with the inductive assumption,
whereas the second case implies  the precondition of the next stage.
Since the precondition of  stage $t+1$ states a contradiction of the form $U_{1,k_{t+1}} = \emptyset$ 
(because $C_{t}=[t]$ is the set of all colors), it  won't be necessary to execute that stage. 
\\
Now, suppose that the precondition of a stage $s \leq t$ is fulfilled in a block $B(k_s)$
with the color set $C_{s-1}=\{c_1,c_2, \ldots c_{s-1}\}$. Choose some color
$c_s \in U_{1,k_s} \not= \emptyset$ ($c_s \not \in C_{s-1}$ by the  precondition)
and consider  the corresponding guard $g_s$ in the left or 
right wing of column $k_s$. By symmetry it is sufficient to discuss the 
first case $g_s \in W_L(k_s)$ where the three conditions of the left-right rule
hold for all $j \in B_R(k_s)$. Let $J_s=\{j_1,j_2, \dots j_K\}$ be the 
set of all columns of depth $1+sm'$ in  $B_R(k_s)$. Note that this condition implies
for all $j_l \in J_s$ that $d_m(j_l)=d_m(k_s)+m'$ and thus $K=|J_s|=2^{m'-1}$.
The new case inspection applies to the top cells of the rows $j_l \in J_s$: 
\\
  (1): $\ \forall_{l \in  [K]}\ c_s \in U_{1,j_l}$\\
  (2): $\ \exists_{l \in  [K]} \ c_s \not \in U_{1,j_l}$\\
If Case (2) occurs with $c_s \not \in U_{1,j_l}$ for a $j_l \in J_s$ then $c_s \not \in U_{1,j}$ for all $j \in B(j_l)$
by condition (c) of the left-right rule. This immediately implies the precondition for the next
stage with $k_{s+1}=j_l$ and $C_{s}=  \{c_1,c_2, \ldots c_{s}\}$. \\
If Case (1) occurs, we consider the polygon $P_s$ formed by the union of all cells
$R_{i,j}$ with $j \in  B_R(k_s)$ and $2+(s-1)m' \leq i \leq 1+sm'$. 
As discussed above the cell decomposition of the polygon $P_s$ is isomorphic to that
of $S_{m'}$. Moreover extending the set $G_s$ of original guards in $P_s$
by pulling down all guards that sit direcly above $P_s$  onto the top edge of 
$P_s$, we obtain a cf-guarding of $P_s$ with  $t-1$ colors, because 
 $c_s$ can't occur as a color in the extended guard set $G_s^+$.
This contradicts the inductive assumption and finishes the proof.   
\end{proof}

Any attempt to adapt  this proof to cf-guardings of $S_m$
with respect to line
visibility encounters the following  problems. \\
{\bf  Problem 1}: It is impossible to subdivide the polygon into a finite set of
visibility cells such that any two points
in a cell would have the same visibility polygon.\\
Solution: The guard set watching a given cell $R_{i,j}$ is replaced by
the guard set watching a  single special
point in the cell.
We always choose the midpoint $p_{i,j}$ of the lower side of $R_{i,j}$.\\
{\bf  Problem 2}: Guards from the left wing of a column $k$ can possibly see points
in the right wing that are much deeper than
$d_m(k)$.\\
Solution:
 The heights of  rows in $S_m$ will be stretched in an
appropriate way such that no
guard from the left wing of a column $k$ can see a special  point $p_{i,j}$ with
$j \in B_R(k)$ and $i > d_m(k)$.\\
{\bf  Problem 3}: A guard that sits deeper than $d_m(k)$ in  $B_L(k)$ can possibly watch
points in column $k$ and even points in $B_R(k)$.\\
Solution: One can't avoid this, but the stretching of  rows will assure that it won't see
points in $B_{RR}(k)$.
It turns out that the left-right rule must be relaxed in such a way that
conditions (a), (b) and (c)
do not hold in whole opposite half block, but they hold (in slightly
modified form) at least in a quarter
subblock. Formally we will refer to this fact by a quantified formula of
the type
$\exists_{ XY \in \{LL,LR,RL,RR\}} \ \forall_{ j \in B_{XY}} \ldots$\\
{\bf  Problem 4}: Pulling a guard to another position like in the construction of
the guard set $G_1^{+}$ for the subpolygon $P_1$
changes the visibility range of the guards and might result in a guard
set that does not cover the same subpolygon it
covered before.\\
Solution: Any conflict-free guarding of $S_m$ will be translated into  purely
combinatorial properties of
the corresponding multicolor tableau, such that
concrete  guard positions don't play any role in the subsequent lower bound proof.

\subsection{Stretched spike polygons and {\bf $t$}-conform tableaux}
\ \ \ \ For the purpose of forcing similar properties for l-visibility as we
used for r-visibility
we introduce a vertically stretched version ${S_m^{\downarrow}}$ of  $S_m$ with the
following geometric properties:
\begin{itemize}
\item The width of each column is $1$ and hence the total width of  
$S_m^{\downarrow}$ is
$2^m-1$.   
\item We  distinguish between combinatorial and geometric depth of a
column:
While  $d_m(k)=m-\pi_2(k)$ is still used for the combinatorial depth,  we
want  the  geometric depth to  be  $d^{\downarrow}_m(k)=
2^{(d_m(k)-1)m}$.
Therefore the height of the first row is $h_1=1$ and the height of
the $i$-th row  $h_i= 2^{im}-2^{(i-1)m}$.
\end{itemize}
Consider the decomposition of $S_m^{\downarrow}$ into r-visibility cells
$R_{i,k}$ and let
$p_{i,k}$ be the midpoint at the bottom side of  $R_{i,k}$.
If $\gamma:G \rightarrow [t]$ for guard set  $G $ is a  conflict-free l-guarding of
$S_m^{\downarrow}$, then
let $M_{i,k}^{\downarrow}$ be the multiset of all colors of guards that see
$p_{i,k}$ and ${\mathcal
M}^{\downarrow}(\gamma)=\left\{M_{i,k}^{\downarrow} \ | \ k \in [2^m-1], i
\in [d_m(k)] \right\}$
the corresponding multicolor tableau.

The following two observations establish  similarities between 
l-visibility in   $S_m^{\downarrow}$
and r-visibility in $S_m$ and substitute Lemmata \ref{rVis} and \ref{uniqueInterval}.
\begin{lemma}
\label{lVis1}
Let $g$ be a guard  in $S_m^{\downarrow}$, $k$ a column
of this polygon with
combinatorial depth $d=d_m(k)$ and geometric depth
$d_m^{\downarrow}(k)=2^{(d-1)m}$.
If $g \in W_R(k)$ ($g \in W_L(k)$)
then $g$ can't see any point $p$ at depth $ d^{\downarrow}(p) \geq
2^{dm}$ in the left (right) block
of $k$, especially $g$ can't see
any point $p_{i,j}$ with $j \in B_L(k)$ ($j \in B_R(k)$) and $i > d$.
\end{lemma}
\begin{proof}
By symmetry it is sufficient to study the first case with $g \in
W_R(k)$, $d^{\downarrow}(p) \geq 2^{dm} $ and
$p$ a point in the subpolygon  $B_L(k)$.
Let $q_L$ be the left vertex of the horizontal polygon edge in column $k$ and
consider the  slopes $s_1$ and $s_2$ of the lines $\overline{pq_L}$ and
$\overline{q_Lg}$.
Since the width of $B_L(k)$ is $2^{m-d}-1$ and
$d^{\downarrow}(p)-d^{\downarrow}(q_L) \geq 2^{dm} - 2^{(d-1)m} =(2^m-1)
\cdot 2^{(d-1)m} $  we get
\[s_1 \geq \frac{(2^m-1) \cdot 2^{(d-1)m}}{2^{m-d}-1}
=\frac{(2^m-1) \cdot 2^{(d-1)m}}{2^{-d}(2^m-2^d) }
> \frac{2^{(d-1)m}}{2^{-d}}= 2^{(d-1)m+d}  \]
Since $g$ is in the right wing of $k$ it is at least one half unit right
of $q_L$ and it is at most
$d_m^{\downarrow}(k)=2^{(d-1)m}$ higher than $q_L$
\[s_2 \leq \frac{2^{(k-1)m}}{1/2}  = 2^{(d-1)m+1}\leq  2^{(d-1)m+d} \]
Thus, $s_1>s_2$ what shows that the corner at $q_l$ blocks the
l-visibility between $p_{i,j}$ and $g$.
\end{proof}
\begin{lemma}
\label{lVis2}
Let $g$ be an l-guard watching a point $p_{i,k} \in S_m^{\downarrow}$.
Then for all $i' \leq i$ and for all $j \in B_L(k)$ or for all $j \in
B_R(k)$
$g$ sees also   $p_{i',j}$.
\end{lemma}
\begin{proof}
Let  $d^{\downarrow}(g)$ be the geometric
depth of  $g$ in $S_m^{\downarrow}$.
\\
{\em Case 1:} If $g$ is even an r-guard for $p_{i,k}$ the claim follows
for all  $j \in B(k)$
by Lemma \ref{rVis}. Otherwise there are two more cases, namely that
$d^{\downarrow}(g)$ is strictly smaller or strictly larger than
$2^{(i-1)m}$ (see Figure \ref{three_guard_pos}).
\\
{\em Case 2:} $d^{\downarrow}(g) < 2^{(i-1)m}$, i.e., $g$ sees  $p_{i,k}$
from above. If
 $g \in W_R(k)$ then $g$ can see all  $p_{i,j}$ with $j \in B_L(k)$ beause
the line segments $p_{i,j}p_{i,k}$ and  $p_{i,k}g$ are contained in
$S_m^{\downarrow}$
and they form a chain that is convex from above.
If $g \in W_L(k)$ then $g$ can see all  $p_{i,j}$ with $j \in B_R(k)$ beause
the line segments $gp_{i,k}$ and  $p_{i,k}p_{ij}$ are contained in
$S_m^{\downarrow}$
and they form a chain that is convex from above. Moreover it is clear
that in $S_m^{\downarrow}$
any guard that sees a point  $p_{i,j}$ will see also all points directly
above, especially
the points $p_{i',j}$ with $i'<i$.
\\
{\em Case 3:} $d^{\downarrow}(g)> 2^{(i-1)m }$, i.e. $g$ sees  $p_{i,k}$
from below.  
Since  $d^{\downarrow}(g) \leq d_m^{\downarrow}(k)$ would imply case 1,
we can additionally
assume  $d^{\downarrow}(g) > d_m^{\downarrow}(k)$, i.e. $g$
is in a cell
$R_{i',j'}$ with $i' \geq d_m(k)$ and $j' \in B_L(k)$ or $j' \in B_R(k)$.
Now we can apply Lemma \ref{lVis1} with $p_{i,k}$ in the role of the
guard $g$
and $q$ in the role of a point $p$ watched by $g$. It turns out  that
$d^{\downarrow}(g) < 2^{d_m(k)m}$, i.e., $g$ lies in row $d_m(k)+1$ of
$S_m^{\downarrow}$.
It follows that depending whether $g$ lies in $B_L(k)$ or $B_R(k)$
it sees all $p_{i,j}$ with $j \in B_L(k)$ or $j \in B_R(k)$ (and all
points directly above as well).
\end{proof}

A  tableau ${\mathcal M}^{\downarrow}(\gamma)$ is in 
standard form if it has  $m$ rows and $N=2^m-1$
columns. But by various constructions, for example restricting it to a single block, one creates  
a tableau having $m$ rows and $N'=2^{m'}-1$ columns for some $m'<m$.
The following definition of $t$-conformity specifies some necessary, but not sufficient conditions a multicolor tableau has if it stems 
from a conflict-free $t$-coloring  of a stretched spike polygon. The advantage is that $t$-conformity is preserved when acting on the tableau with 
operations defined below.

\noindent
\begin{definition}
\label{t-admis}
Let  $m' \leq m$ be natural numbers and $N'=2^{m'}-1$.
A scheme of  multisets over the set $[t ]$ of the form
${\mathcal M}=\left(M_{i,k} \, | \, k \in [N'],  i \in [d_m(k)] \right)$
is called a  
$t$-conform $(m \times N')$-multicolor tableau
if the following properties hold:
\begin{enumerate}
\item $\forall_{ k \in [N']} \ \forall_{   i \in [d_m(k)]} \ U_{i,k} \not=
\emptyset$.
\item $\forall_{ k \in [N']} \ \forall_{  1 \leq i < i'  \leq d_m(k)} \
M_{i',k} \subseteq M_{i,k}$.
\item $\forall_{ k \in [N']}\ \forall_{   i \in [d_m(k)]} \ \forall_{  c \in
U_{i,k}} \
\exists_{ XY \in \{LL,LR,RL,RR\}} \ \forall_{ j \in B_{XY}(k)} \quad Q(c,k,j)$ \\
where the predicate  $Q(c,k,j) $ is the conjunction of 
three conditions:
\begin{enumerate}
\item $c \in M_{i,j}$
\item  $c \in U_{i,j} \rightarrow  c \not \in M_{d_m(k)+2,j}$  
\item $ c \not \in U_{i,j} \rightarrow \exists_{ Z \in \{L,R\}} \ \forall_
{j' \in B_Z(j)}\, c \not \in U_{d_m(k),j'}$.
\end{enumerate}
\end{enumerate}
\end{definition}

Note the first two properties are identical with
those 
in Proposition \ref{r-visAdmiss}.
The third one, however,  is a proper relaxation of the left-right rule  there.
Thus, any tableau ${\mathcal M}(\gamma)$ resulting from a conflict-free
r-guarding
of $S_m$ with $t$ colors is also $t$-conform.
\begin{proposition} \label{l-visAdmiss}
The multicolor tableau  ${\mathcal M}^{\downarrow}(\gamma)$ for a
conflict-free
l-guarding of the polygon $S_m^{\downarrow}$ with $t$ colors is
$t$-conform.
\end{proposition}
\begin{proof}
There is nothing to prove for the uniqueness condition and for the
monotonicity.
Now let us assume $c \in U_{i,k}$ with a
corresponding guard $g$.
By symmetry we may suppose $g \in W_R(k)$. Like in  Lemma
\ref{lVis2}
there are three cases to distinguish (see Figure \ref{three_guard_pos}):
\begin{enumerate}
\item $p_{i,k}$ is r-visible from $g$.
\item $p_{i,k}$ is not r-visible from $g$ and $p_{i,k}$ is deeper than $g$.
\item  $p_{i,k}$ is not r-visible from $g$ and $g$ is deeper than $p_{i,k}$.
\end{enumerate}
\begin{figure}
\centering
\includegraphics[scale=.5]{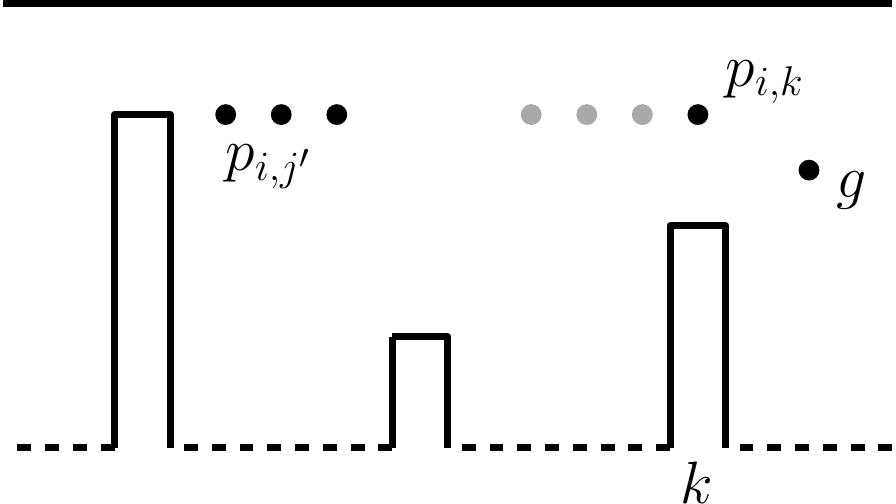}\quad\quad
\includegraphics[scale=.5]{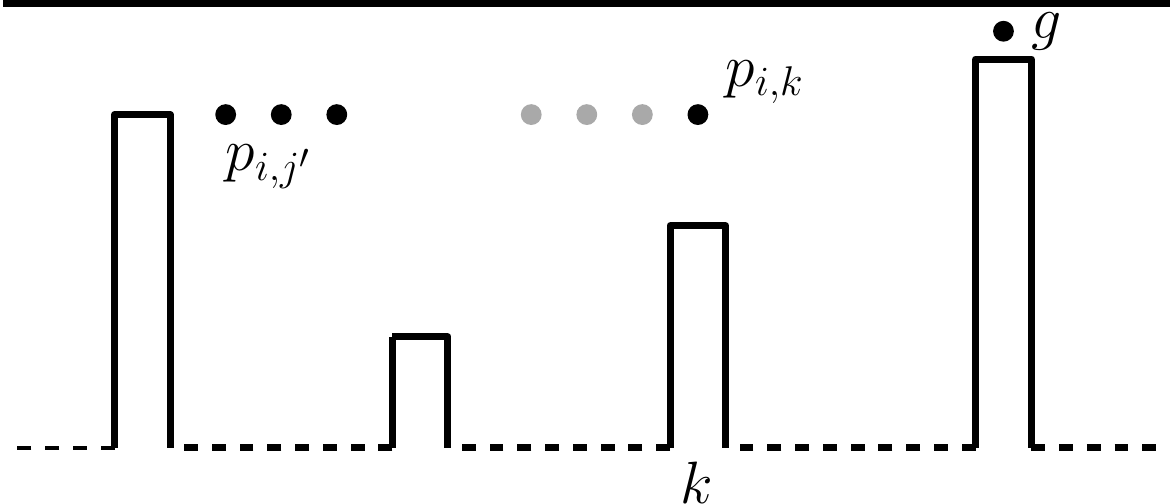}
\includegraphics[scale=.5]{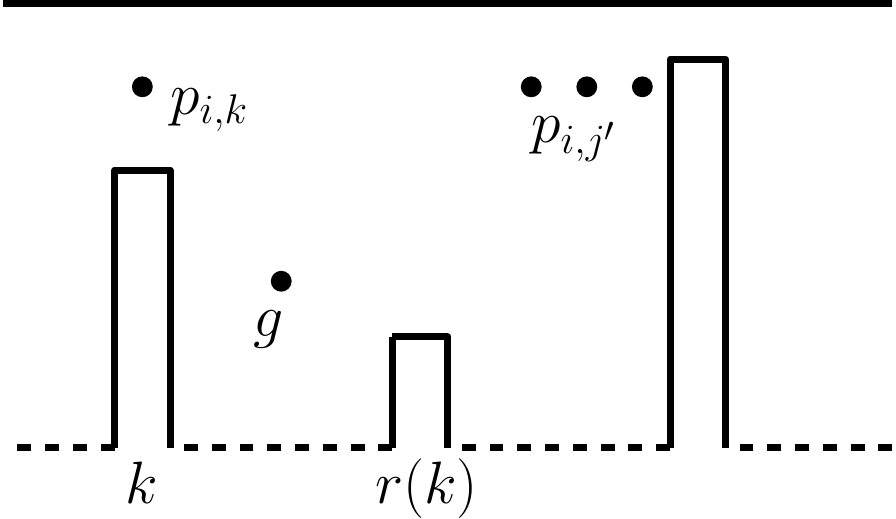}
\caption{Possible guard positions  with respect to the
point $p_{i,k}$. Note that it is impossible to
  display the exponential growth of the row heights in the drawing.}
\label{three_guard_pos}
\end{figure}
In Case 1 and Case 2 we choose $XY=LL$ (but $XY=LR$ would also work -
the gray points).
In Case 3 the choice depends on the position of $g$ relative to the central
column $r(k)$ of the block $B_R(k)$:
\[XY= \left\{ \begin{array}{ll}
  RL & \mbox{ if } g \in W_R(r(k)) \\
  RR & \mbox{ if } g \in W_L(r(k)). \\
\end{array} \right.
\]
It remains to establish the three conditions of $Q(c,k,j)$ for all
$j \in B_{XY}(k)$. Condition (a) is obvious in case 1 and case 2. In
case 3 it follows from
the fact that $g$ can't be deeper than $d_m^{\downarrow}(r(k))$ (see
Case 3 in the proof of
Lemma \ref{lVis2}).
\\
For condition (b) suppose that  $c \in U_{i,j}$. This implies that $g$
is the only guard with
color $c$ that sees $p_{i,j}$. However in all three
cases $g$ is in the wing opposite  to  block $B_{XY}(k)$ and then $g$
can't see any point of
combinatorial depth ${d_m(k)+2}$ in $B_{XY}(k)$ by Lemma \ref{lVis1}.
It's worth observing that
depth ${d_m(k)+1}=d_m(r(k))$  would not suffice in case 3.
However, any other guard with color $c$ watching $p_{d_m(k)+2,j}$ would also watch 
$p_{i,j}$ and
contradicts the uniqueness of $g$. Thus $ c \not \in M_{d_m(k)+2.j}$.
\\
Finally, let us suppose $c \not \in U_{i,j}$, then there is a second
guard $g'$
for  $p_{i,j}$. Now  we can conclude from  Lemma \ref{lVis2} that $g'$
watches
all  points $p_{i,j'}$ for   $j' \in B_L(j)$ or for all $j' \in B_R(j)$.
This proves condition (c).
\end{proof}

\begin{proposition} \label{TabConstr}
If ${\mathcal M}=\left(M_{i,k} \, | k\in [N'] ,  i\in [d_m(k)] \right)$
is a $t$-conform  $(m \times N')$-multicolor tableau with $N'=2^{m'}-1$ for some $m'\leq m$. Then the following three constructions
yield new   $t$-conform tableaux  ${\mathcal M}_1,{\mathcal M}_2,{\mathcal M}_3$ :
\begin{enumerate}
\item ${\mathcal M}_1$ is the restriction of ${\mathcal M}$ to a block
$B(k)$;
\item  ${\mathcal M}_2$ results from  deleting  the top $m-m'$ rows
of  ${\mathcal M}$;
\item ${\mathcal M}_3$ results from  selecting  $2^{m^*}-1$ columns
for some $m^*<m'$
with respect to the following rules:
\begin{itemize}
\item For all even $k \in [2^{m^*}-1]$ choose column $k\cdot 2^{m'-m^*}$
of ${\mathcal M}$ as
column $k$ of ${\mathcal M}_3$.
\item For all odd $k \in [2^{m^*}-1]$ choose any column  $j$   of
${\mathcal M}$ with
$(k-1) \cdot 2^{m'-m^*} < j < (k+1) \cdot 2^{m'-m^*} $, delete from that
column all entries
of depth $d> m^*+m-m'$ and use this truncated column as column $k$ of
${\mathcal M}_3$.
\end{itemize}
\end{enumerate}
\end{proposition}
\begin{proof}
Recall, the width of $B(k)$ is  $N^*=2^{m^*}-1$ where  
$m^*= 2\pi_2(k)$.
So the only thing that has to do  for ${\mathcal M}_1$ is  shifting  the column
numbering
from the interval $B(k)=[k-2^{\pi_2(k)}+1, k+2^{\pi_2(k)}-1]$ to
$[N^*]$. Then  ${\mathcal M}_1$ is t-conform.
\\
For the second construction it is sufficient to shift down the indices
of all undeleted
rows by $m-m'$. Then ${\mathcal M}_2$ is an $m'\times N'$ tableau. Note
that an old row index
$d_m(k)=m - \pi_2(k)$ becomes $d_{m'}(k)$. Having that in mind, it
is also trivial
that  ${\mathcal M}_2$ is $t$-conform.
\\
The construction of ${\mathcal M}_3$ already contains the renumbering of
indices.  Again, it
is not hard to conclude the $t$-conformity because the construction
inherits the relations
of being a column in the left (or right) subblock of another column.
\end{proof}

\begin{theorem}
\label{l-cf-lowerbound}
$\chi_{cf}^l (n) \in
\Omega\left(\frac{\log{\log{n}}}{\log{\log{\log{n}}}}\right)$.
\end{theorem}
\begin{proof}
Despite similarities to the proof of Theorem
\ref{r-cf-lowerbound}
some essential modifications have to be implemented.
The function $m(t)$  is now
defined by $m(1)=3$ and
$m(t)=1+t \cdot (m(t-1)+1)$ for
 $t \geq 2$.

\noindent
{\em Claim:} An $m(t) \times (2^{m(t)}-1)$-tableau cannot be $t$-conform.

\noindent
The inequality $m(t) \leq
(t+1)!$ is no longer valid in general,
but, it still holds for all $t \geq 5$. In fact, $m(5)=651<720=(5+1)!$
and the  induction
step works for any $t \geq 6$ as follows:
\\
$m(t) = t\cdot(m(t-1)+1)+1 \leq t(t!+1)+1 = t \cdot t! +(t+1) \leq t \cdot t!
+t!=(t+1)!$
\\
Hence using Proposition \ref{l-visAdmiss} one can then deduce the theorem
from the claim like before.

\noindent
In the  proof of the claim  by induction on $t$  the induction base for
$t=1$ works
with similar arguments as before. Any $1$-conform $3 \times 7$ tableau
requires to set $U_{i,k}=\{1\}$ for all $k \in [7]$ and all $i \in
[d_3(k)]$.
However, applying property 3 to the situation $1 \in U_{1,4}$ 
yields a contradiction with condition (c).

The induction step is proved by contradiction again. Assume 
that  there are no $(t-1)$-conform
$m'\times N'$-tableaux with $m'=m(t-1)$ and $N'=2^{m'}-1$, but there  is
a $t$-conform
$m\times N$-tableau $\mathcal M$ for $m=m(t)$ and $N=2^{m}-1$.
The proof consists of $s \leq t$ stages. The precondition of stage $s$
is the existence of a $t$-conform $m \times N_{s-1}$-tableau
%${\mathcal M}_{s-1}$
where $N_{s-1}=2^{m-(s-1)(m'+1)}-1$ and the additional property that there is a
set $C_{s-1} \subset [t]$
consisting of $s-1$ colors, such that for all $c \in C_{s-1}$ and for all
$k \in [N_s]$ holds $c \not \in U_{1,k}$. The precondition for the first
stage is given by  
$\mathcal M$ with $N_0=N$ and $C_0= \emptyset$, but $\mathcal M$ will
change after every stage.
The postcondition  of the $s$-th stage
is either a contradiction obtained by constructing a   $(t-1)$-conform
$m'\times N'$-tableau (the stop condition, case 1) or the creation of
the precondition for the next step
(case 2). Note that if the stop condition did not occur after the $t$-th
stage, then
the new precondition gives also a contradiction because
$C_t=[t]$ and $N_t=2^1-1=1$, i.e.,~it would result in  a $t$-conform
$m \times 1$-tableau (a single column) such that no color can be  unique
in $M_{1,1}$.
\\
Now suppose that an  $m \times N_{s-1}$-tableau $\mathcal M$ with a
color set  $C_{s-1}$ fulfills   the
precondition for stage $s$  with $ 1\leq s \leq t$.
Let $k=\frac{N_{s-1}+1}{2}$ be the central column of ${\mathcal M}$ and
$c_s \in U_{1,k}$.
Note that the precondition implies $c_s \not \in  C_{s-1}$. Then by 
property 3 of $t$-conform
tableaux there is some $XY \in \{LL,LR,RL,RR\}$ such that predicate
$Q(c_s,k,j)$ is true
for all $j \in B_{XY}(k)$. Again we subdivide the block $B_{XY}(k)$
into $K=2^{m'-1}$ subblocks of equal width. These subblocks can be
defined by their central columns
$j_l$ where $l \in [K]$. Note that their width  just fits to the
precondition of the next stage because $B_{XY}(k)$ has width
$\frac{N_{s-1}+1}{4}-1$ and
consequently all  $B(j_l)$  have width:
  \[ \frac{N_{s-1}+1}{4 \cdot 2^{m'-1} }-1 =  \frac{2^{m-(s-1)(m'+1)}}{4
\cdot 2^{m'-1} }-1
= \frac{2^{m-(s-1)(m'+1)}}{2^{m'+1} }-1= 2^{m-s(m'+1)}-1 \]
Due to the weaker conditions encoded in predicate $Q(c,k,j)$ we have to modify the case inspection:\\
(1) $\forall_{l \in [K]} \ \exists_{ j' \in B(j_l)} \ c_s \in U_{1,j'}$\\
(2) $\exists_{ l \in [K]} \ \forall_{ j' \in B(j_l)} \ c_s \not\in U_{1,j'}$\\
In Case 1 we can immediately derive a contradiction  using the
constructions of Proposition
\ref{TabConstr}: First we restrict (the current) $\mathcal M$
to the block $B_{XY}(k)$, then we use the column selection with $m^*=m'$
where
the even columns (numbered $2l$ for $l \in [K]$) of the new tableau are
the ones that separate the subblocks
$B_{j_l}$ and $B_{j_{l+1}}$ from each other and the odd columns $2l-1$
are chosen from $B_{j_l}$
with respect to the property $c_s \in U_{1,j'}$. Supposing that  $c_s$
is not unique in the top set of
an even column would contradict condition (c) of
predicate $Q(c_s,k,j)$.
Thus $c_s$ is unique everywhere in the first row of the new tableau and
with respect to condition (b)
it does not occur at all in third row or deeper. Each column of this
new  tableau ${\mathcal M}'$
has depth $d \geq 3$ because all columns of ${\mathcal M}'$ had been
selected from a quarter
subblock $B_{XY}(k)$. Now we apply construction 2 (deletion of top rows)
to ${\mathcal M}'$ to obtain an
$m' \times N'$-tableau ${\mathcal M}^*$.  This way at least the two top
rows of ${\mathcal M}'$
are deleted and thus color $c_s$ doesn't occur anymore in ${\mathcal M}^*$.
Finally, we can  replace color $t$ by  color $c_s$ to obtain a
$(t-1)$-conform
$m' \times N'$-tableau.
\\
Case 2 is now the easier one because replacing $\mathcal M$ by a block
$B(j_l)$ such that
$\forall_{ j' \in B(j_l)} \ c_s \not\in U_{1,j'}$ (construction 1) yields
the precondition for the next
stage with $C_s=C_{s-1} \cup \{c_s\}$.
\end{proof}

% ---------------------------------------------------------------------------------------------------
% Conclusions
% ---------------------------------------------------------------------------------------------------
\section{Conclusions}
\ \ \ \ We have shown almost tight bounds for the chromatic AGP for orthogonal 
simple polygons if based on r-visibility. While the upper bound proofs use known techniques, we consider the multicolor tableau method
 for the lower bounds to be the main technical contribution of our paper. This method seems to be unnecessarily complicated for the lower
 bound on $\chi_{cf}^r(n)$. But it shows its strength when applied to the line visibility case. It is this discrete structure which enables
 one to apply induction. Otherwise we would not know how to show a lower bound for a continuum of possible  guard positions with strange dependencies 
plus all possible colorings.\\
We conjecture that indeed  $\chi_{cf}^r(n)\in \Omega (\log\log n)$ using spike polygons and this should also yield a $\log\log n$ lower bound for the line visibility case via the stretched version.
But one cannot hope for more, $\log\log n$ is also an upper bound for cf-guarding of stretched spike polygons using line visibility. To improve this lower bound one has to look for other polygons. 
%---------------------------- Bibliography -------------------------------

%---------------------------- Bibliography -------------------------------

\bibliographystyle{splncs}

\end{document}